\documentclass[11pt]{article}
\usepackage{vmargin, cite}
\setpapersize{USletter}
\setmarginsrb{1.2in}{1.2in}{1.2in}{1.2in}{0mm}{0mm}{5mm}{5mm}

\usepackage{mystyle}
\usepackage{graphicx}
\usepackage{amsmath}

\title{Greedy Selfish Network Creation\footnote{See \cite{L12} for the original publication.}\\{\small (Full Version)}}

\author{Pascal Lenzner\thanks{Department of Computer Science, Humboldt-Universit\"at zu Berlin, Unter den Linden 6, 10099 Berlin, Germany. Email: \email{lenzner@informatik.hu-berlin.de}}}

\date{}

\begin{document}

\maketitle

\begin{abstract}
  \noindent We introduce and analyze \emph{greedy equilibria} (GE) for the well-known model of selfish network creation by Fabrikant et al.~[PODC'03]. GE are interesting for two reasons: (1) they model outcomes found by agents which prefer smooth adaptations over radical strategy-changes, (2) GE are outcomes found by agents which do not have enough computational resources to play optimally. In the model of Fabrikant et al. agents correspond to Internet Service Providers which buy network links to improve their quality of network usage. It is known that computing a best response in this model is NP-hard. Hence, poly-time agents are likely not to play optimally. But how good are networks created by such agents? We answer this question for very simple agents. Quite surprisingly, naive greedy play suffices to create remarkably stable networks. Specifically, we show that in the \textsc{Sum} version, where agents attempt to minimize their average distance to all other agents, GE capture Nash equilibria (NE) on 
  trees and that any GE is in $3$-approximate NE on general networks. For the latter we also provide a lower bound of $\tfrac{3}{2}$ on the approximation ratio. For the \textsc{Max} version, where agents attempt to minimize their maximum distance, we show that any GE-star is in $2$-approximate NE and any GE-tree having larger diameter is in $\tfrac{6}{5}$-approximate NE. Both bounds are tight. We contrast these positive results by providing a linear lower bound on the approximation ratio for the \textsc{Max} version on general networks in GE. This result implies a locality gap of $\Omega(n)$ for the metric min-max facility location problem, where $n$ is the number of clients.
\end{abstract}

\section{Introduction}
The area of Network Design is one of the classical and still very active fields in the realm of Theoretical Computer Science and Operations Research. 
But there is this curious fact: One of the most important networks which is increasingly shaping our everyday life -- the Internet -- cannot be fully explained by classical Network Design theory. 
Unlike centrally designed and optimized networks, the Internet was and still is created by a multitude of selfish agents (Internet Service Providers (ISPs)), who control and modify varying sized portions of the network structure (``autonomous systems''). This decentralized nature is an obstacle to approaching the design and analysis of the Internet as a classical optimization problem. Interestingly, each agent does face classical Network Design problems, i.e. minimizing the cost of connecting the \emph{own} network to the rest of the Internet while ensuring a high quality of service. The Internet is the result of the interplay of such local strategies and can 
be considered as an equilibrium state of a game played by selfish agents.  

The classical and most popular solution concept of such games is the (pure) Nash equilibrium~\cite{Nash}, which is a stable state, where no agent unilaterally wants to change her current (pure) strategy. 
However, Nash equilibria (NE) have their difficulties. Besides their purely descriptive, non-algorithmic nature, there are two problems: (1) With NE as solution concept agents only care if there is a better strategy and would perform radical strategy-changes even if they only yield a tiny improvement. (2) In some games it is computationally hard to even tell if a stable state is reached because computing the best possible strategy of an agent is hard. Thus, for such games NE only predict stable states found by supernatural agents. 
 
But what solutions are actually found by more realistic players, i.e. by agents who prefer smooth strategy-changes and who can only perform polynomial-time computations? And what impact on the stability has this transition from supernatural to realistic players?   

In this paper, we take the first steps towards answering these questions for one of the most popular models of selfish network creation. This model, called the \textsc{Network Creation Game} (NCG), was introduced a decade ago by Fabrikant et al.~\cite{Fab03}. In NCGs agents correspond to ISPs who create links towards other ISPs while minimizing cost and maximizing their quality of network usage. It seems reasonable that ISPs prefer greedy refinements of their current strategy (network architecture) over a strategy-change which involves a radical re-design of their infrastructure. Furthermore, computing the best strategy in NCGs is NP-hard. Hence, it seems realistic to assume that agents perform smooth strategy-changes and that they do not play optimally. We take this idea to the extreme by considering very simple agents and by introducing and analyzing a natural solution concept, called \emph{greedy equilibrium}, for which agents can easily compute whether a stable state is reached and which models an ISP's 
preference for smooth strategy-changes. 

\subsection{Model and Definitions}
In NCGs~\cite{Fab03} there is a set of $n$ agents $V$ and each agent $v\in V$ can buy an edge $\{v,u\}$ to any agent $u\in V$ for the price of $\alpha> 0$. Here $\alpha$ is a fixed parameter of the game which specifies the cost of creating any link. The strategy $S_v$ of an agent $v$ is the set of vertices towards which $v$ buys an edge. Let $G = (V,E)$ be the induced network, where an edge $\{x,y\}\in E$ is present if $x\in S_y$ or $y\in S_x$. The network $G$ will depend heavily on the parameter $\alpha$. To state this explicitly, we let $(G,\alpha)$ denote the network induced by the strategies of all agents $V$. In a NCG agents selfishly choose strategies to minimize their cost. There are basically two versions of NCGs, depending on the definition of an agent's cost-function. In the \emph{\textsc{Sum} version}~\cite{Fab03}, agents try to minimize the sum of their shortest path lengths to all other nodes in the network, while in the \emph{\textsc{Max} version}~\cite{De07}, agents try to minimize their 
maximum shortest path distance to any other network node. The precise definitions are as follows: Let $S_v$ denote agent $v$'s strategy in $(G,\alpha)$, then we have for the \textsc{Sum} version that the cost of agent $v$ is 
$c_v(G,\alpha) = \alpha |S_v| + \sum_{w\in V(G)} d_G(v,w)$, if $G$ is connected and $c_v(G,\alpha) = \infty$, otherwise. 
For the \textsc{Max} version we define agent $v$'s cost as 
$c_v'(G,\alpha) = \alpha |S_v| + \max_{w\in V(G)} d_G(v,w)$, if $G$ is connected and $c_v'(G,\alpha) = \infty$, otherwise.
In both cases $d_G(\cdot,\cdot)$ denotes the shortest path distance in the graph $G$. Note that both cost functions nicely incorporate two conflicting objectives: Agents want to pay as little as possible for being connected to the network while at the same time they want to have good connection quality. 
For NCGs we are naturally interested in networks where no agent unilaterally wants to change her strategy. Clearly, such outcomes are pure NE and we let \textsc{Sum}-NE denote the set of all pure NE of NCGs for the \textsc{Sum} version and \textsc{Max}-NE denotes the corresponding set for the \textsc{Max} version.  

Another important notion is the concept of \emph{approximate Nash equilibria}. Let $(G,\alpha)$ be any network in a NCG. 
For all $u \in V(G)$ let $c(u)$ and $c^*(u)$ denote agent $u$'s cost induced by her current pure strategy in $(G,\alpha)$ and by her best possible pure strategy, respectively. 
We say that $(G,\alpha)$ is a $\beta$-approximate Nash equilibrium if for all agents $u\in V(G)$ we have $c(u) \leq \beta c^*(u)$, for some $\beta \geq 1$. 

\subsection{Related Work}
The work of Fabrikant et al.~\cite{Fab03} did not only introduce the very elegant model described above. 
Among other results, the authors showed that computing a best possible strategy of an agent is NP-hard. 

To remove the quite intricate dependence on the parameter $\alpha$, Alon et al.~\cite{ADHL10} recently introduced the \textsc{Basic Network Creation Game} (BNCG), in which a network $G$ is given and agents can only ``swap'' incident edges to decrease their cost. Here, a swap is the exchange of an incident edge with a non-existing incident edge. The cost of an agent is defined like in NCGs but without the edge-cost term.
The authors of~\cite{ADHL10} proposed the \emph{swap equilibrium}~(SE) as solution concept for BNCGs. A network is in SE, if no agent unilaterally wants to swap an edge to decrease her cost. This solution concept has the nice property that agents can check in polynomial time if they can perform an improving strategy-change. The greedy equilibrium, which we analyze later, can be understood as an extension of the swap equilibrium which has similar properties but provides agents more freedom to act.
Note, that in BNCGs an agent can swap \emph{any} incident edge, whereas in NCGs only edges which are bought by agent $v$ can be modified by agent $v$. This problem, first observed by Mihal\'ak and Schlegel~\cite{MS10}, can easily be circumvented, as recently proposed by the same authors in~\cite{MS12}: BNCGs are modified such that every edge is owned by exactly one agent and agents can only swap \emph{own} edges. The corresponding stable networks of this modified version are called \emph{asymmetric swap equilibrium}~(ASE). However, independent of the ownership, edges are still two-way. 
These simplified versions of NCGs are an interesting object of study since (asymmetric) swap equilibria model the local weighing of decisions of agents and despite their innocent statement they tend to be quite complicated structures. 
In~\cite{L11} it was shown that greedy dynamics in a BNCG converge very quickly to a stable state if played on a tree. The authors of~\cite{CHKS12} analyzed BNCGs on trees with agents having communication interests. 
However, simplifying the model as in~\cite{ADHL10} is not without its problems.  
Allowing only edge-swaps implies that the number of edges remains constant. Hence, this model seems too limited to explain the creation of rapidly growing networks. 

A part of our work focuses on tree networks. Such topologies are common outcomes of NCGs if edges are expensive, which led the authors of~\cite{Fab03} to conjecture that all (non-transient) stable networks of NCGs are trees if $\alpha$ is greater than some constant. The conjecture was later disproved by Albers et al.~\cite{Al06} but it was shown to be true for high edge-cost. In particular, the authors of~\cite{MS10} proved that all stable networks are trees if $\alpha > 273n$ in the \textsc{Sum} version or if $\alpha > 129$ in the \textsc{Max} version. Experimental evidence suggests that this transition to tree networks already happens at much lower edge-cost and it is an interesting open problem to improve on these bounds. 

Demaine et al.~\cite{De09} investigated NCGs, where agents cannot buy every possible edge. Furthermore, Ehsani et al.~\cite{Ehs11} recently analyzed a bounded-budget version. Both versions seem realistic, but in the following we do not restrict the set of edges which can be bought or the budget of an agent. 
Clearly, such restrictions reduce the qualitative gap between simple and arbitrary strategy-changes and would lead to weaker results for our analysis. Note, that this indicates that outcomes found by simple agents in the edge or budget-restricted version may be even more stable than we show in the following sections.    

To the best of our knowledge, approximate Nash equilibria have not been studied before in the context of selfish network creation. Closest to our approach here may be the work of Albers et al.~\cite{AL10}, which analyzes for a related game how tolerant the agents have to be in order to accept a centrally designed solution. We adopt a similar point of view by asking how tolerant agents have to be to accept a solution found by greedy play.

Guyl\'as et al.~\cite{GKSB12} recently published a paper having a very similar title to ours. They investigate networks created by agents who use the length of ``greedy paths'' as communication cost and show that the resulting equilibria are substantially different to the ones we consider here. Their term ``greedy'' refers to the distances whereas our term ``greedy'' refers to the behavior of the agents.   

\subsection{Our Contribution}
We introduce and analyze greedy equilibria (GE) as a new solution concept for NCGs. This solution concept is based on the idea that agents (ISPs) prefer greedy refinements of their current strategy (network architecture) over a strategy-change which involves a radical re-design of their infrastructure. Furthermore, GE represent solutions found by very simple agents, which are computationally bounded.
We show in Section~\ref{sec_greedyintro} that such greedy refinements can be computed efficiently and clarify the relation of GE to other known solution concepts for NCGs. 

Our main contribution follows in Section~\ref{sec_quality} and Section~\ref{sec_quality_max}, where we analyze the stability of solutions found by greedily playing agents. For the \textsc{Sum} version we show the rather surprising result that, despite the fact that greedy strategy-changes may be sub-optimal from an agent's point of view, \textsc{Sum}-GE capture \textsc{Sum}-NE on trees. That is, in any tree network which is in \textsc{Sum}-GE \emph{no} agent can decrease her cost by performing \emph{any} strategy-change. For general networks we prove that any network in \textsc{Sum}-GE is in $3$-approximate \textsc{Sum}-NE and we provide a lower bound of $\frac{3}{2}$ for this approximation ratio. Hence, we are able to show that greedy play almost suffices to create perfectly stable networks.  

For the \textsc{Max} version we show that these games have a strong non-local flavor which yields diminished stability. Here even GE-trees may be susceptible to non-greedy improving strategy-changes. Interestingly, susceptible trees can be fully characterized and we show that their stability is very close to being perfect. Specifically, we show that any GE-star is in $2$-approximate \textsc{Max}-NE and that any GE-tree having larger diameter is in $\tfrac{6}{5}$-approximate \textsc{Max}-NE. We give a matching lower bound for both cases. For non-tree networks in GE the picture changes drastically. We show that for GE-networks having a very small $\alpha$ the approximation ratio is related to their diameter and we provide a lower bound of $4$. For $\alpha \geq 1$, we show that there are non-tree networks in \textsc{Max}-GE, which are only in $\Omega(n)$-approximate \textsc{Max}-NE. The latter result yields that the locality gap of uncapacitated metric min-max facility 
location is in $\Omega(n)$.    

Regarding the complexity of deciding Nash-stability, we show that there are simple polynomial time algorithms for tree networks in both versions. Furthermore, greedy-stability represents an easy to check certificate for $3$-approximate Nash-stability in the \textsc{Sum} version.
         
\section{Greedy Agents and Greedy Equilibria}\label{sec_greedyintro}
We consider agents which check three simple ways to improve their current infrastructure. The three operations are
\begin{itemize}
 \item \emph{greedy augmentation}, which is the creation of \emph{one} new own link,
 \item \emph{greedy deletion}, which is the removal of \emph{one} own link,
 \item \emph{greedy swap}, which is a swap of \emph{one} own link. 
\end{itemize}
Computing the best augmentation/deletion/swap for one agent can be done in $\mathcal{O}(n^2(n+m))$ steps by trying all possibilities and re-computing the incurred cost. 
Observe, that these smooth strategy-changes induce some kind of organic evolution of the whole network which seems highly adequate in modeling the Internet. This greedy behavior naturally leads us to a new solution concept:
\begin{definition}[Greedy Equilibrium]
$(G,\alpha)$ is in \emph{greedy equilibrium} if no agent in $G$ can decrease her cost by buying, deleting or swapping \emph{one} own edge.
\end{definition}
Note, that GE can be understood as solutions which are obtained by a distributed local search procedure performed by selfish agents.

The next theorem relates GE to other solution concepts in the \textsc{Sum} version. See Fig.~\ref{fig:classes} for an illustration. Relationships are similar in the \textsc{Max} version.
\begin{theorem}\label{thm_structure}
 For the \textsc{Sum} version it is true that NE~$\subset$~GE~$\subset$~ASE and that SE~$\subset$~ASE. Furthermore, we have NE$\setminus$SE~$\neq \emptyset$, GE$\setminus$SE~$\neq \emptyset$, $($GE$\setminus$SE$)\setminus$NE~$\neq \emptyset$, $($GE$\setminus$NE$)~\cap$ SE ~$\neq \emptyset$ and NE ~$\cap$ GE ~$\cap$ SE ~$\neq \emptyset$.
\end{theorem}
\begin{figure}[ht]
  \centering
  \includegraphics[width=14cm]{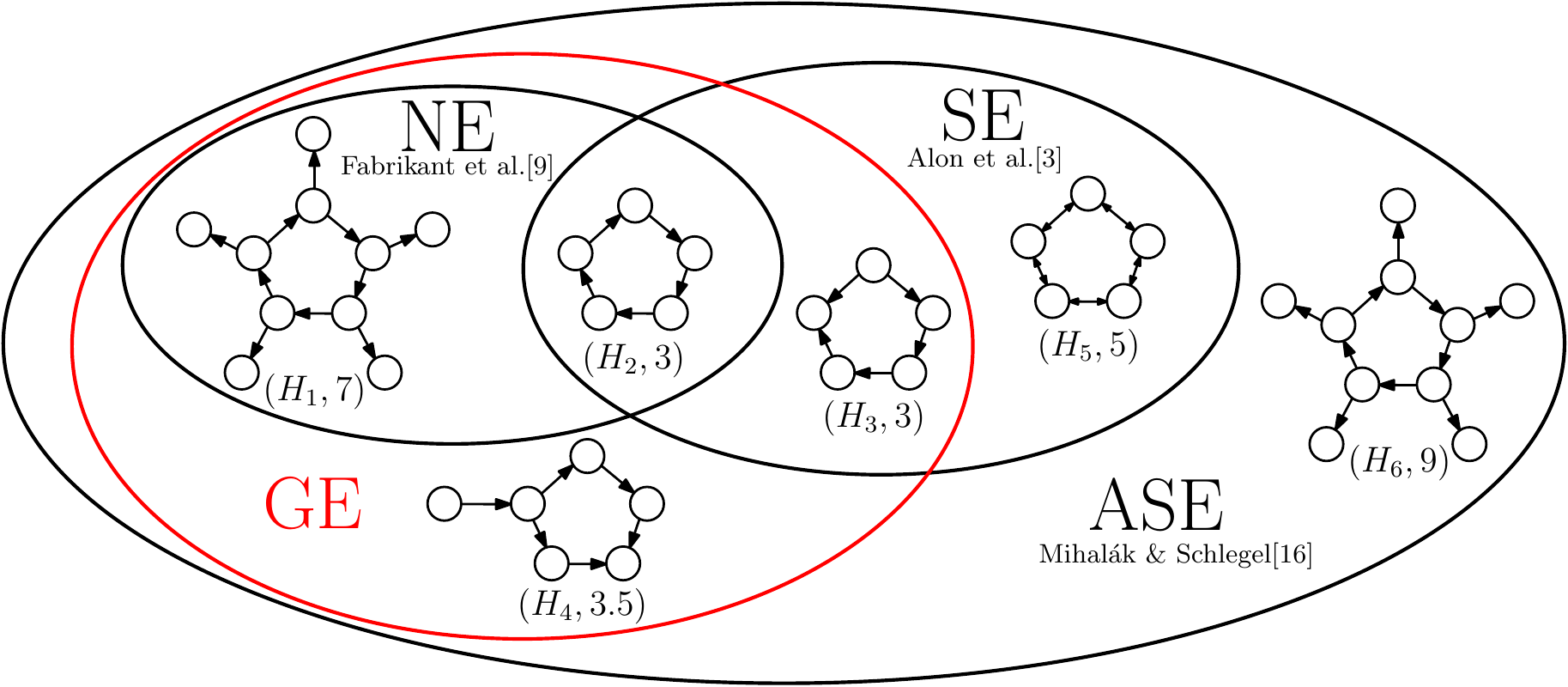}
  \caption{Relations between solution concepts for NCGs in the \textsc{Sum} version. Edge-directions indicate edge-ownership, edges point away from its owner.}
  \label{fig:classes}
\end{figure}
\begin{proof}
 It is easy to see that \textsc{Sum}-NE~$\subseteq$ \textsc{Sum}-ASE and \textsc{Sum}-SE $\subseteq$ \textsc{Sum}-ASE must hold, since in both cases, we restrict the set of available strategies for the agents. Clearly, greediness restricts the possible strategies of an agent as well. Hence, we have \textsc{Sum}-NE~$\subseteq$ \textsc{Sum}-GE. Furthermore, by the same argument, if no agent can buy, delete or swap one own edge, then such a network must be in directed swap equilibrium. It follows that \textsc{Sum}-GE~$\subseteq$ \textsc{Sum}-ASE.
 
 Consider the networks depicted in Fig.~\ref{fig:classes}. It follows from Lemma~$6$ in \cite{ADHL10} that any graph having diameter $2$ is in \textsc{Sum}-SE. Thus, $(H_2,3),(H_3,3),(H_5,5) \in \text{\textsc{Sum}-SE}$, since the edge-cost parameter and the edge-ownerships have no influence on the stability in a BNCG. Observe, that $(H_1,7),(H_6,9)\notin \text{\textsc{Sum}-SE}$, since any leaf-agent can swap her edge towards a neighbor of another leaf-agent and thereby strictly decrease her cost. Furthermore, $(H_4,3.5) \notin \text{\textsc{Sum}-SE}$, since an agent having distance $3$ towards the leaf-agent can decrease her cost by swapping an edge towards the neighbor of the leaf. 

 Now we show that a cycle having $5$ vertices, $C_5$ for short, is in \textsc{Sum}-GE for any edge assignment if and only if $1\leq \alpha \leq 4$. Since $C_5$ is in \textsc{Sum}-SE, we only have to show that no agent wants to buy or delete one edge if and only if $1\leq \alpha \leq 4$. Since every vertex of $C_5$ has eccentricity $2$, we have that an agent can strictly decrease her cost by buying one edge if and only if $\alpha <1$. On the other hand, since deleting one edge increases the distance-cost of the moving agent by $4$, we have that an agent can strictly decrease her cost by deleting one edge if and only if $\alpha > 4$. Thus, $(H_2,3),(H_3,3)\in \text{\textsc{Sum}-GE}$ and $(H_5,5)\notin \text{\textsc{Sum}-GE}$. 

 Next, we show that $H_1=H_6$ is in \textsc{Sum}-GE for $6\leq \alpha \leq 8$: For $H_1$ to be in \textsc{Sum}-GE, we have to make sure that no leaf-vertex can buy an edge and that every cycle-vertex has bought the best possible edge and cannot be better off by removing that edge or by purchasing one \emph{additional} edge. By symmetry of the construction, we can focus on one leaf-vertex $l$ only. A best possible edge for $l$ is the edge towards a cycle-vertex which has maximum distance to $l$. This edge decreases agent $l$'s distance-cost by $6$. Now we consider a cycle-vertex $u$ and again, by symmetry, it suffices to argue for vertex $u$. Observe, that agent $u$ cannot remove her edge towards the neighboring leaf, since this would disconnect the network. If $x$ removes an edge towards a neighboring cycle-vertex, then agent $u$'s distance-cost increases by $8$. Furthermore, it is easy to see that no edge-swap can decrease agent $u$'s cost. Hence, it remains to show that agent $u$ cannot buy an additional 
edge and thereby strictly decrease her cost. A best possible additional edge for $u$ is an edge towards a non-neighboring cycle-vertex, which yields a distance decrease of $2$. Hence, if $\alpha > 2$, no such additional edge will be bought by $u$. Analogously, it is easy to check that $H_4$ is in \textsc{Sum}-GE for $3\leq \alpha \leq 4$. If we restrict agents only to swapping own edges, it follows that $(H_6,9)\in \text{\textsc{Sum}-ASE} \setminus \text{\textsc{Sum}-GE}$. 

 Now, let us investigate $(H_1,7)$. Note, that agents of $H_1$ who do not own any edge cannot change their strategy to strictly decrease their cost. Hence, we only have to argue that no cycle-vertex $u$ can unilaterally change her strategy to strictly decrease her cost. By symmetry of the construction, if suffices to argue for agent $u$. Let $l_u$ be $u$'s leaf-neighbor and let $w$ be $u$'s cycle-neighbor to which $u$ owns an edge and let $v$ be $u$'s other cycle-neighbor. Observe, that $u$ has to buy the edge towards $l_u$ in any strategy to ensure connectedness. Hence, we can safely ignore this edge. Furthermore, since $\alpha \leq 8$, removing edge $\{u,w\}$ does not yield a strict cost decrease for $u$. Since $(H_1,7)$ is a greedy equilibrium, we have that $u$ cannot swap edge $\{u,w\}$ with some other edge to decrease her cost. It remains to show that $u$ cannot strictly decrease her cost by removing edge $\{u,w\}$ and buying at least two edges. First, let us assume that $u$ can remove edge $\{u,w\}$ 
and simultaneously buy two edges $\{u,x\}$ and $\{u,y\}$ and thereby strictly decrease her cost. Observe, that $x\neq w$ and $y\neq w$ must hold, since otherwise the edge not connecting to $w$ would be a greedy augmentation. Furthermore, it is easy to show that $v\neq x$, $v\neq y$ and $x\neq y$ must hold and that $x$ and $y$ cannot be leaves, since leaves are always dominated by their corresponding cycle-neighbors. Hence, the only possible strategy for agent $u$, which satisfies the mentioned constraints, is to buy the edges $\{u,z_1\}$ and $\{u,z_2\}$, where $z_1$ and $z_2$ are the cycle-vertices which have maximum distance to $u$. This strategy yields a distance decrease of $10$ compared to buying only edge $\{u,l_u\}$. Clearly, every edge in an equilibrium strategy, which is not required for ensuring connectedness of the network must yield at least a distance decrease of $\alpha$, since otherwise the agent would be better off without buying that edge. Since $10 < 2\alpha$, we have that agent $u$'s new 
cost is strictly higher than $u$'s cost in $(H_1,7)$. Observe, that removing $\{u,w\}$ and buying three edges $\{u,x\},\{u,y\},\{u,z\}$, with $x\neq w, y\neq w$ and $z\neq w$, yields a distance decrease of $12 < 3\alpha$. Hence, agent $u$ cannot strictly decrease her cost by buying $3$ edges. For more than three edges, where no edge is allowed to connect to $w$, an analogous argument yields that $u$ cannot strictly decrease her cost. Hence, agent $u$ cannot change her strategy to strictly decrease her cost. Analogously, it is easy to check that $(H_2,3)\in \text{\textsc{Sum}-NE}$.

 The network $(H_3,3) \notin \text{\textsc{Sum}-NE}$, since the vertex which owns two edges can strictly decrease her cost by removing both edges and buying one edge towards a vertex in distance $2$ in $H_3$. Finally, $(H_4,3.5)\notin \text{\textsc{Sum}-NE}$, since the agent who owns two edges can strictly decrease her cost by performing a similar strategy change as the respective agent in $(H_3,3)$. 
\end{proof}

\section{The Quality of Sum Greedy Equilibria}\label{sec_quality}
This section is devoted to discussing the quality of greedy equilibrium networks in the \textsc{Sum} version. 
We begin with a simple but very useful property.  
\begin{lemma}\label{lemma_augmentation}
  If an agent $v$ cannot decrease her cost by buying \emph{one} edge in the \textsc{Sum} version, then buying $k>1$ edges cannot decrease agent $v$'s cost.  
\end{lemma}
\begin{proof}
  Let $v$ be an agent who cannot strictly decrease her cost in network $(G,\alpha)$ by purchasing one edge. Let $q$ denote the number of edges in $(G,\alpha)$ owned by agent $v$. Now assume towards a contradiction that agent $v$ can strictly decrease her cost by purchasing $k>1$ edges $e_1,\dots,e_k$. Let $(G^k,\alpha)$ be the network $(G,\alpha)$ augmented by these $k$ edges. Let $c_v(G,\alpha)$ and $c_v(G^k,\alpha)$ denote agent $v$'s cost in $(G,\alpha)$ and $(G^k,\alpha)$, respectively. Hence, we have $c_v(G^k,\alpha) < c_v(G,\alpha)$. Let $D$ and $D^k$ denote agent $v$'s distance-cost in $(G,\alpha)$ and $(G^k,\alpha)$, respectively. We have $c_v(G,\alpha) = q\alpha + D$ and $c_v(G^k,\alpha) = q\alpha + k\alpha + D^k$. Let $(G^1,\alpha)$ denote the network $(G,\alpha)$, where agent $v$ has built the best possible additional edge $e^*$. That is, there is no other additional edge $e'$, such that agent $v$ can strictly decrease her cost by swapping edge $e^*$ with edge $e'$. Since $(G,\alpha)$ is in greedy 
equilibrium, we have $c_v(G^1,\alpha) = q\alpha + \alpha + D^1 \geq c_v(G,\alpha)$, where $c_v(G^1,\alpha)$ denotes agent $v$'s cost in the network $(G^1,\alpha)$ and $D^1$ denotes $v$'s distance-cost in $(G^1,\alpha)$.
  Hence, we have
  \begin{equation}c_v(G^k,\alpha) < c_v(G,\alpha) \iff k\alpha < D - D^k \label{eqkalpha}\end{equation}
  and
  \begin{equation}c_v(G^1,\alpha) \geq c_v(G,\alpha) \iff \alpha \geq D - D^1. \label{eq1alpha}\end{equation}
  Let $g^k = D-D^k$ and $g^1 = D-D^1$, that is, $g^k$ and $g^1$ denote the distance decrease of agent $v$ by building edges $e_1,\dots,e_k$ simultaneously or by building edge $e^*$, respectively. 
  For all edges $l \in \{e_1,\dots,e_k,e^*\}$, let $g^l$ denote the decrease in distance-cost for agent $v$ if only the edge $l$ is inserted into network $(G,\alpha)$. Observe, that $g^k \leq g^{e_1} + g^{e_2}+\dots + g^{e_k}$ holds.

  By inequality (\ref{eqkalpha}) we have that $\alpha < \frac{g^k}{k} \leq \frac{g^{e_1} + g^{e_2}+\dots + g^{e_k}}{k}$. It follows that $\alpha < g^{e_j}$, for some $1\leq j\leq k$, since otherwise we would have $\alpha < \frac{g^{e_1} + g^{e_2}+\dots + g^{e_k}}{k}\leq \frac{k\alpha}{k} = \alpha$. Furthermore, since $e^*$ is the best possible additional edge for agent $v$ in $(G,\alpha)$, we have $g^{e_j} \leq g^{e^*}$. It follows that $\alpha < g^{e^*}$, which contradicts inequality~(\ref{eq1alpha}). 
\end{proof}

\subsection{Tree Networks in Sum Greedy Equilibrium}
We show that in a NCG all stable trees found by greedily behaving agents are even stable against \emph{any} strategy-change. Hence, in case of a tree equilibrium \emph{no} loss in stability occurs by greedy play. This is a counter-intuitive result, since for each agent alone being greedy is clearly sub-optimal (the network in Fig.~\ref{fig:tree_1} with $\alpha = 6$ is an example). Thus, the following theorem shows the emergence of an optimal outcome out of a combination of sub-optimal strategies.
\begin{theorem}\label{thm_tree_eq}
 If $(T,\alpha)$ is in \textsc{Sum}-GE and $T$ is a tree, then $(T,\alpha)$ is in \textsc{Sum}-NE.
\end{theorem}
\noindent Before we prove Theorem~\ref{thm_tree_eq}, we first provide some useful observations. The well-known notion of a \emph{$1$-median}~\cite{KH79} is used: 
  A \emph{$1$-median} of a connected graph $G$ is a vertex $x\in V(G)$, where $x \in \arg\min_{u\in V(G)}\sum_{w\in V(G)}d(u,w).$
\begin{lemma}\label{lem_tree_median}
 Let $(T,\alpha)$ be a tree network in \textsc{Sum}-GE. If agent $u$ owns edge $\{u,w\}$ in $(T,\alpha)$, then $w$ must be a $1$-median of its tree in the forest $T-\{u\}$.
\end{lemma}
\begin{proof}
 Assume towards a contradiction that vertex $w$ is not a $1$-median vertex in its respective tree $T_w$ in the forest $T-\{u\}$. Clearly, agent $u$'s unique shortest paths to all vertices in $V(T_w)$ in $(T,\alpha)$ traverse vertex $w$ and we have that agent $u$'s distance-cost within $T_w$ is $|V(T_w)| + \sum_{v \in V(T_w)} d(w,v)$. Let $x$ be a $1$-median vertex of $T_w$. By definition, we have that $\sum_{v\in V(T_w)}d(x,v) < \sum_{v\in V(T_w)}d(w,v)$. Thus, agent $u$ can strictly decrease her distance-cost within $T_w$ by performing an edge-swap from $w$ towards $x$. This contradicts the fact that $(T,\alpha)$ is in \textsc{Sum}-GE.
\end{proof}
\noindent Let $(T,\alpha)$ be any tree network in \textsc{Sum}-GE and let $T^u$ be the forest induced by removing all edges owned by agent $u$ from $T$. Let $F^u$ be the forest $T^u$ without the tree containing vertex $u$. The above lemma directly implies the following:
\begin{corollary}\label{cor_greedy}
 Let $(T,\alpha)$ be in \textsc{Sum}-GE, and let $F^u$ be defined as above. Agent $u$'s strategy in $(T,\alpha)$ is the optimal strategy among all strategies that buy exactly one edge into each tree of $F^u$. 
\end{corollary}
\noindent Let $x\in V(T)$ be a $1$-median of the tree $T$. Let $u\notin V(T)$ be a special vertex. We consider the network $(G_T^u,\alpha)$, which is obtained by adding vertex $u$ and inserting edge $\{u,x\}$, which is owned by $u$, in $T$ and by assigning the ownership of all other edges arbitrarily among the respective endpoints of any other edge in $G_T^u$. Furthermore, let $y_1,\dots,y_l$ denote the neighbors of vertex $x$ in $T$ and let $T_{y_i}$, for $1 \leq i \leq l$, denote the maximal subtree of $T$ which is rooted at $y_i$ and which does not contain vertex $x$. See Fig.~\ref{fig:tree_1}~(left) for an illustration.
\begin{figure}[!h]
  \centering
  \includegraphics[width=10cm]{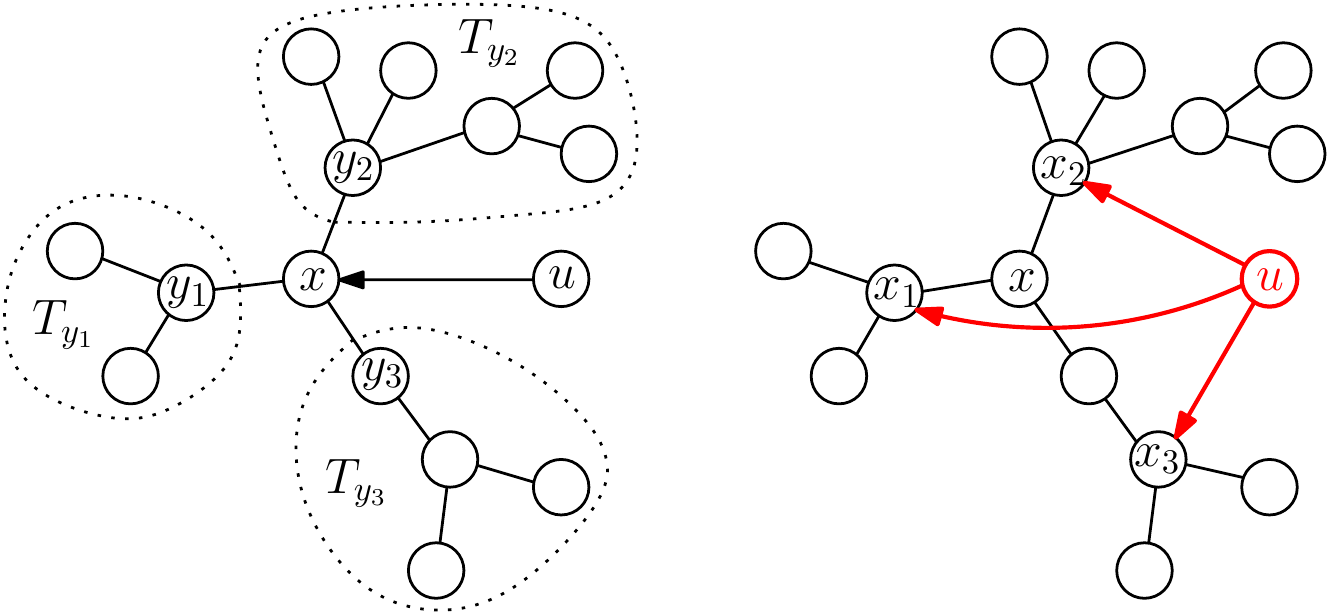}
  \caption{The network $(G_T^u,\alpha)$ before and after agent $u$ changes her strategy to $S_u^*$.}
  \label{fig:tree_1}
\end{figure}
We consider a special strategy of agent $u$ in $(G_T^u,\alpha)$: Let $S_u^* = \{x_1,\dots,x_k\}$ be the best strategy of agent $u$ which purchases \emph{at least two} edges. The situation with agent $u$ playing strategy $S_u^*$ is depicted in Fig.~\ref{fig:tree_1}~(right).
\begin{lemma}\label{lem_different_subtrees}
 Let $(G_T^u,\alpha)$, $S_u^* = \{x_1,\dots,x_k\}$ and the subtrees $T_{y_i}$, for $1\leq i \leq l$ be specified as above. There is no subtree $T_{y_i}$, which contains all vertices $x_1,\dots,x_k$.
\end{lemma}
\begin{proof}
 We assume towards a contradiction that there is a strategy $S_u^*$ buying $k>1$ edges and a subtree $T_{y_i}$ of $T$ such that $x_1,\dots,x_k$ are vertices of $T_{y_i}$. 

 We claim that if this is the case, then there is a strategy $S_u'$ which purchases exactly $k$ edges and which strictly outperforms strategy $S_u^*$, that is, agent $u$ can strictly decrease her cost by switching from strategy $S_u^*$ to strategy $S_u'$. Clearly, this yields a contradiction to $S_u^*$ being the best strategy for agent $u$.  

 Consider the vertices $x_1,\dots,x_k$ induced by strategy $S_u^*$. By assumption, all these vertices are contained in subtree $T_{y_i}$. Observe, that vertex $x$ does not belong to any subtree $T_{y_j}$. We will use the following well-known fact~\cite{KH79} about a $1$-median in a tree stated in our terminology: Vertex $x$ is a $1$-median of tree $T$ having $n$ vertices if and only if $|V(T_{y_i})| \leq \frac{n}{2}$ for all $1\leq i \leq l$. 

 Let $x' \in \{x_1,\dots,x_k\}$ be the vertex having minimum distance to vertex $x$ and let $x''$ be the neighbor of $x'$ which is closer to $x$. (Note, that $x''=x$ is possible and that $x''$ must be a non-neighbor of $u$.) Clearly, we have $d_T(x',x)\geq 1$. Let $S_u'$ be the strategy $S_u^*$ with only one modification: Vertex $x'$ is replaced by vertex $x''$. 

 We claim that $S_u'$ yields strictly less cost for agent $u$ than strategy $S_u^*$. Observe, that since $x$ is a $1$-median  we have $|V(T)\setminus V(T_{y_i})| \geq \frac{n}{2}$. Hence, the replacement of $x'$ by $x''$ yields a cost decrease for agent $u$ by at least $\frac{n}{2}$. On the other hand, this replacement increases agent $u$'s cost by at most $\frac{n}{2}-1$. This is true, because $k>1$ and $d_T(x',x'') = 1$ we have that agent $u$'s distances to all but one vertices in $T_{y_i}$ can possibly increase by $1$. Since we have only replaced $x'$ by $x''$ all other distances stay the same. Hence, we have that $S_u'$ yields strictly less cost for agent $u$ than strategy $S_u^*$ and we have a contradiction.
\end{proof}
\noindent Next, let us consider two special strategies of agent $u$. Let $S_u^1$ be agent $u$'s best strategy, which buys at least two edges including one edge towards vertex $x$. Furthermore, let $S_u^2$ be agent $u$'s best strategy, which buys at least two edges, but no edge towards vertex $x$. 
\begin{lemma}\label{lem_tree_no_leaf}
 Let $(G_T^u,\alpha)$, $S_u^1$, $S_u^2$ and vertex $x$ be specified as above. Let $x_j \in S_u^2$ be a vertex which has minimum distance to $x$ among all vertices in $S_u^2$. If strategy $S_u^2$ yields less cost for agent $u$ than strategy $S_u^1$, then $x_j$ cannot be a leaf of $G_T^u$.
\end{lemma}
\begin{proof}
 We assume towards a contradiction that $S_u^2$ yields strictly less cost for agent $u$ than strategy $S_u^1$ and vertex $x_j \in S_u^2$, which has minimum distance to $x$ among all vertices in $S_u^2$, is a leaf of $G_T^u$. Let $x_j'$ be the unique neighbor of $x_j$. It follows that $d(x_j',x) = d(x_j,x) -1$. There are two cases: 
 
 If $d(x_j,x)\geq 2$, then let $S_u'$ be the strategy $S_u^2$, where vertex $x_j$ is replaced by vertex $x_j'$. We claim that agent $u$ can strictly decrease her cost by switching from strategy $S_u^2$ to strategy $S_u'$. Observe, that by switching from $S_u^2$ to $S_u'$, agent $u$ decreases her distance to $x$ and to $x_j'$ by one. On the other hand, only the distance to vertex $x_j$ increases by one. Observe, that $|S_u'| = |S_u^2|$ and $x\notin S_u'$. Hence, $S_u'$ yields strictly less cost for agent $u$ than strategy $S_u^2$, which is a contradiction to the fact that $S_u^2$ is agent $u$'s best strategy which buys at least two edges and no edge towards $x$. 

 On the other hand, consider the case where $d(x_j,x) \leq 1$. Note, that $x_j\neq x$, since $x\notin S_u^2$. Hence, we have $d(x_j,x) = 1$. Let $S_u''$ be the strategy $S_u^2$, where we replace $x_j$ by $x$. We have $|S_u^2| = |S_u''|$ and $x\in S_u''$. Furthermore, since $x\notin S_u^2$, $d(x_j,x) = 1$ and since $x_j$ is a leaf, we have that strategy $S_u''$ yields at most the cost of $S_u^2$ for agent $u$. This contradicts the fact that $S_u^2$ strictly outperforms $S_u^1$. 
\end{proof}
\noindent Now we have all the tools we need to prove Theorem~\ref{thm_tree_eq}.
\begin{proof}[Proof of Theorem~\ref{thm_tree_eq}]
 We will prove the contra-positive statement of Theorem~\ref{thm_tree_eq}. We show that if an agent $u$ can decrease her cost by performing a strategy-change in a tree network $(T,\alpha)$ which is in \textsc{Sum}-GE, then there is an agent $z$ in $V(T)$ who can decrease her cost by performing a \emph{greedy} strategy-change. In that case we have a contradiction to $(T,\alpha)$ being in \textsc{Sum}-GE.\\ 
 If agent $u$ can decrease her cost by buying, deleting or swapping one own edge, then we have $u=z$ and we are done. Hence, we assume that agent $u$ cannot decrease her cost by a greedy strategy-change but by performing an arbitrary strategy-change. 
 We consider agent $u$'s strategy-change towards the best possible \emph{arbitrary} strategy $S^*$ (if $u$ has more than one such strategy, then we choose the one which buys the least number of edges). Clearly, agent~$u$ cannot remove any owned edge without purchasing edges, since $T$ is a tree and the removal would disconnect $T$. Furthermore, since $(T,\alpha)$ is in \textsc{Sum}-GE and by Lemma~\ref{lemma_augmentation}, agent $u$ cannot decrease her cost by purchasing $k>0$ \emph{additional} edges. Hence, the only way agent $u$ can possibly decrease her cost is by removing $j$ own edges and building $k$ edges simultaneously. Clearly, $k\geq j$ must hold. Furthermore, by Corollary~\ref{cor_greedy}, it follows that $k>j$. 
 Let $F^u$ be the forest obtained by removing the $j$ edges owned by agent $u$ from $T$ and let $T^*$ be the tree in $F^u$ which contains vertex $u$. Observe that among the $k$ new edges, there cannot be edges having an endpoint in $T^*$. This is true because $(T,\alpha)$ is in \textsc{Sum}-GE and by Lemma~\ref{lemma_augmentation}. Any such edge would be a possible greedy augmentation which we assume not to exist. Hence, by the pigeonhole principle, we have that there must be at least one tree $T_q$ in $F^u$ into which agent $u$ buys at least \emph{two} edges with strategy $S^*$. We focus on $T_q$ and will find agent $z$ within.\\
 Let $\{u,x\}$, with $x\in V(T_q)$, be the unique edge of $T$ which connects $u$ to the subtree~$T_q$. Hence, agent~$u$'s strategy-change to $S^*$ removes edge $\{u,x\}$ and buys $k_q>1$ edges $\{u,x_1\},\dots,\{u,x_{k_q}\}$, with $x_j \in V(T_q)$ for $1\leq j \leq k_q$. Let $X = \{x_1,\dots,x_{k_q}\}$. By Lemma~\ref{lemma_augmentation}, we have $x_j \neq x$, for $x_j \in X$.
 Let $y_1,\dots,y_l$ denote the neighbors of vertex $x$ in $T_q$ and let $T_{y_1},\dots,T_{y_l}$ be the maximal subtrees of $T_q$ not containing vertex~$x$, which are rooted at vertex~$y_1,\dots,y_l$, respectively. Let $x_a \in X$ be a vertex of $X$ which has minimum distance to vertex~$x$. Let $T_a \in \{T_{y_1},\dots,T_{y_l}\}$ be the subtree containing $x_a$. By Lemma~\ref{lem_different_subtrees}, we have that there is a subtree $T_b \in \{T_{y_1},\dots,T_{y_l}\}$, with $T_b\neq T_a$, which contains at least one vertex of~$X$. Let $B = \{x_{b_1},\dots,x_{b_p}\} = X \cap V(T_b)$. 
Furthermore, since no strategy which buys at least two edges including an edge towards $x$ into $T_q$ outperforms $u$'s greedy strategy within $T_q$ and by Lemma~\ref{lem_tree_no_leaf}, we have that vertex~$x_a$ cannot be a leaf. That is, there is a vertex~$z \in V(T_q)$, which is a neighbor of $x_a$, such that $d(z,x) > d(x_a,x)$. We show that agent $z$ can decrease her cost by buying \emph{one} edge in $(T,\alpha)$.\\
 First of all, notice that by definition of $S^*$, we have that \emph{each} edge $\{u,x_j\}$, with $x_j \in X$, must independently of the other bought edges yield a distance decrease of \emph{more than} $\alpha$ for agent~$u$. Otherwise agent~$u$ could remove this edge and obtain a strictly better (or smaller) strategy, which contradicts the fact that $S^*$ is the best possible strategy (buying the least number of edges). Let $D_j\subset V(T_q)$ be the set of vertices to which edge $\{u,x_j\}$ is the first edge on agent~$u$'s unique shortest path. Since $x_a$ has minimum distance to $x$, it follows that $D_r \subseteq V(T_b)$ for $r\in \{b_1,\dots,b_p\}$. 
 The main observation is that agent $z$ faces in some sense the same situation as agent~$u$ with strategy $S^*$ but without all edges $\{u,y\}$, where $y\in B$: Both have vertex $x_a$ as neighbor and their shortest paths to any vertex in $T_b$ all traverse $x_a$ and $x$. Remember, that each edge $\{u,y\}$, for all $y\in B$, yields a distance decrease of more than $\alpha$ for agent~$u$ and that $D_r \subseteq V(T_b)$, for $r\in \{b_1,\dots,b_p\}$. Furthermore, removing all those edges from $S^*$ yields a strict cost increase for agent $u$. This implies that agent $z$ can decrease her cost by buying all edges $\{z,y\}$, for $y\in B$, simultaneously. If $|B| = 1$, then this strategy-change is a greedy move by agent~$z$ which decreases $z$'s cost. If $|B|>1$, then, by the contra-positive statement of Lemma~\ref{lemma_augmentation}, it follows that there exists \emph{one} edge $\{z,y^*\}$, with $y^* \in B$, which agent~$z$ can greedily buy to decrease her cost.
\end{proof}
\subsection{Non-Tree Networks in Sum Greedy Equilibrium}
There exist non-tree networks in \textsc{Sum}-GE, since, as shown by Albers et al.~\cite{Al06}, there exist non-tree networks in \textsc{Sum}-NE and we have \textsc{Sum}-NE~$\subseteq$ \textsc{Sum}-GE. Having Theorem~\ref{thm_tree_eq} at hand, one might hope that this nice property carries over to non-tree greedy equilibria. Unfortunately, this is not true.
\begin{theorem}\label{thm_approx_lowerbound}
 There is a network in \textsc{Sum}-GE which is not in $\beta$-approximate \textsc{Sum}-NE for $\beta<\frac{3}{2}$.
\end{theorem}

\begin{proof}
 We consider a special family $G_1,G_2,\dots$ of graphs. The graph $G_k$ is constructed as follows: 
 We have $V(G_k) = \{u,w,x,y_1,\dots,y_k\} \cup \{z_i^j \mid 1 \leq i,j \leq k\}$. Vertex $u$ owns edges towards $y_1,\dots,y_k$, vertex $w$ owns edges towards $x$ and $u$, each vertex $z_i^j$ owns an edge to $x$ and $y_i$ and the vertices $y_1,\dots,y_k$ form a clique, with arbitrary edge-ownership. Fig.~\ref{fig:lowerbound}(left) illustrates this construction for $k=3$. 
 \begin{figure}[!ht]
  \centering
  \includegraphics[width=12cm]{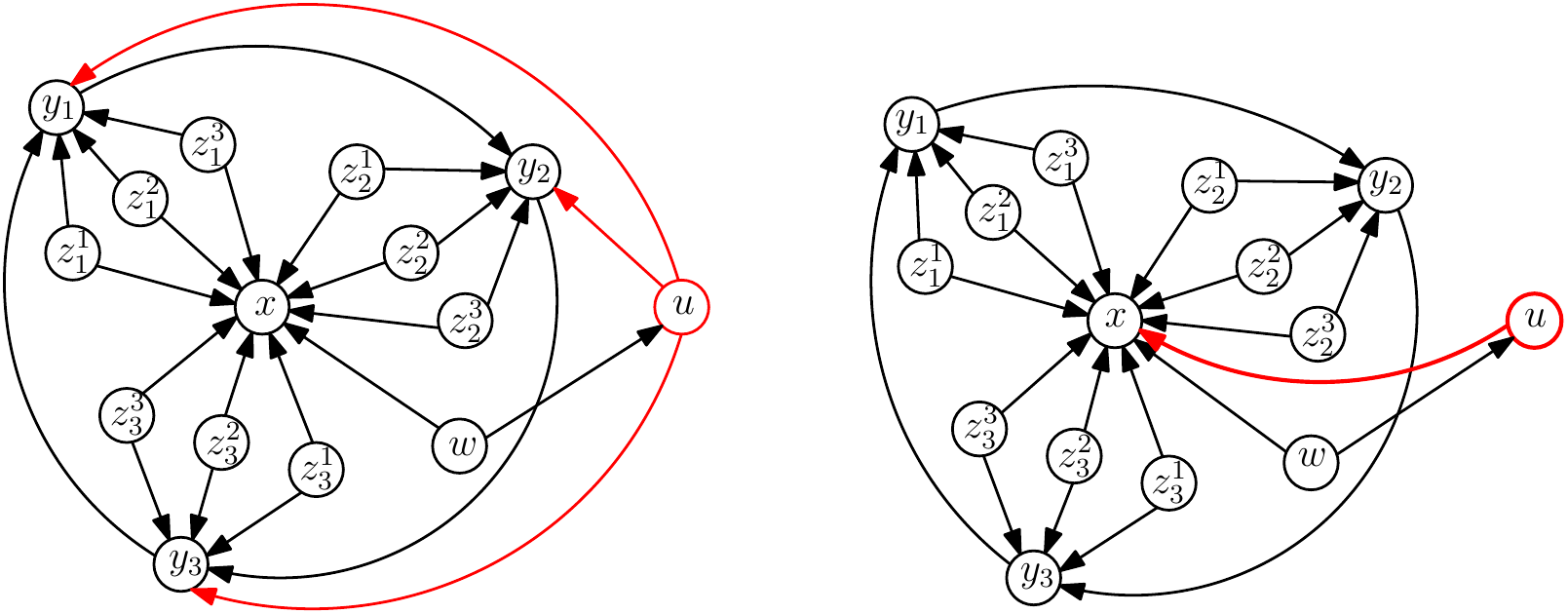}
  \caption{The network $(G_k,k+1)$ for $k=3$ and agent~$u$'s best response. Edges point away from its owner. For $k\to\infty$ agent $u$'s improvement approaches a factor of $\frac{3}{2}$.}
  \label{fig:lowerbound}
\end{figure}
 
 First we show that the network $(G_k,k+1)$ is in \textsc{Sum}-GE and then we show that agent~$u$'s best strategy yields a cost decrease by a factor of roughly~$\frac{3}{2}$.  

 Note, that $G_k$ has diameter $2$. Since $\alpha = k+1 > 1$, it follows that no agent can buy an edge to decrease her cost. Furthermore, swapping any own edge cannot decrease an agent's cost either, since the number of neighbors stays the same. Thus, we only have to argue that no agent can delete an own edge to decrease her cost. Consider agent $u$. Deleting edge $\{u,y_i\}$ increases $u$'s distances to $y_i,z_i^1,\dots,z_i^k$ by one. Hence, for $\alpha = k+1$, this operation does not decrease agent~$u$'s cost. An agent~$y_i$ is in essentially the same situation. If $y_i$ deletes her edge $\{y_i,y_j\}$, then $y_i$'s distances to $y_j,z_j^1,\dots,z_j^k$ increase by one. Thus, agents $y_1,\dots,y_k$ cannot delete an edge to decrease their cost. If agent~$w$ deletes edge $\{w,u\}$, then all distances towards $u,y_1,\dots,y_k$ increase by one. Furthermore if $w$ deletes edge $\{w,x\}$, then all distances towards $x$ and all $z_i^j$, for $1\leq i,j  \leq k$ increase by at least one. Thus, agent~$w$ cannot delete 
an edge to decrease her cost. Finally, consider an agent~$z_i^j$. Deleting edge $\{z_i^j,x\}$ increases $z_i^j$'s distances to $x$ and all $z_p^q$, for $p\neq i$ and $1\leq q \leq k$. Deleting edge $\{z_i^j,y_i\}$ increases $z_i^j$'s distances to $u,y_1,\dots,y_k$ by one. Hence, no agent can delete an edge to decrease her cost and we have that $(G_k,k+1)$ is in \textsc{Sum}-GE.

 Now consider a strategy-change of agent~$u$ from strategy $S_u = \{y_1,\dots,y_k\}$ to strategy $S_u^* = \{x\}$, see Fig.~\ref{fig:lowerbound}(right). Let $(G_k^*,k+1)$ be the network induced by $S_u^*$. We claim that $S_u^*$ is agent~$u$'s best possible strategy. It is easy to see that no other strategy $S_u'$, with $|S_u'|\leq 1$ outperforms $S_u^*$. Furthermore, note that with strategy $S_u^*$ agent~$u$ has exactly the $k$ vertices $y_1,\dots,y_k$ at distance $3$. Any edge $\{u,y_i\}$ yields a cost decrease of exactly $k+1$, but since $\alpha = k+1$, such an edge does not decrease agent~$u$'s cost. Clearly, edges towards a vertex $z_i^j$ are even worse than edges towards $y_i$. By Lemma~\ref{lemma_augmentation}, we have that even more additional edges cannot decrease $u$'s cost. Furthermore, it is easy to see that strategy $S_u$ is agent~$u$'s best possible strategy, which does not buy an edge towards $x$. We will show that $S_u^*$ yields strictly less cost than $S_u$ for agent~$u$, which will settle the 
claim that $S_u^*$ is optimal.   

 Let $c(S_u,k)$ and $c(S_u^*,k)$ denote agent~$u$'s cost in $(G_k,k+1)$ and $(G_k^*,k+1)$, respectively. We have
 $$ \lim_{k \to \infty}\frac{c(S_u,k)}{c(S_u^*,k)} = \lim_{k \to \infty}\frac{k\alpha + k + 1 + 2(k^2 + 1)}{\alpha + 2 + 2k^2 + 3k} = \lim_{k \to \infty}\frac{3k^2 + 2k + 3}{2k^2 + 4k + 3} = \frac{3}{2}.$$
 Thus, for any $\beta < \frac{3}{2}$ there is a $k'$ such that $c(S_u,k') > \beta c(S_u^*,k')$, which implies that the \textsc{Sum}-GE $(G_{k'},k'+1)$ is not a $\beta$-approximate \textsc{Sum}-NE for $\beta < \frac{3}{2}$. 
\end{proof}

 \noindent Now let us turn to the good news. We show that \textsc{Sum}-GEs cannot be arbitrarily unstable. On the contrary, they are very close to \textsc{Sum}-NEs in terms of stability. 
\begin{theorem}\label{thm_approx_upperbound}
 Every network in \textsc{Sum}-GE is in $3$-approximate \textsc{Sum}-NE. 
\end{theorem}
\begin{proof}
 We prove Theorem~\ref{thm_approx_upperbound} by providing a ``locality gap preserving'' reduction to the \textsc{Un\-cap\-aci\-ta\-ted Metric Facility Location} problem (UMFL)~\cite{Vaz01}.\\
 Let $u$ be an agent in $(G,\alpha)$ and let $Z$ be the set of vertices in $V(G)$ which own an edge towards $u$. Consider the network $(G',\alpha)$, where all edges owned by agent~$u$ are removed. Observe, that the set $Z$ is the same in $(G,\alpha)$ and $(G',\alpha)$.
Let $\mathcal{S} = \{U \mid U \subseteq (V(G')\setminus\{u\}) \wedge U \cap Z = \emptyset\}$ denote the set of agent~$u$'s pure strategies in $(G',\alpha)$ which do not induce multi-edges or a self-loop.
 We transform $(G',\alpha)$ into an instance $I(G')$ for UMFL as follows: \\
Let $V(G')\setminus\{u\} = F = C$, where $F$ is the set of facilities and $C$ is the set of clients. For all facilities $f \in Z\cap F$ we define the opening cost to be $0$, all other facilities have opening cost~$\alpha$. Thus, $Z$ is exactly the set of cost $0$ facilities in $I(G')$. For every $i,j \in F \cup C$ we define $d_{ij} = d_{G'}(i,j)+1$. If there is no path between $i$ and $j$ in $G'$, then we define $d_{ij} = \infty$. Clearly, since the distance in $G'$ is metric we have that all distances $d_{ij}$ in $I(G')$ are metric as well. See Fig.~\ref{fig:umfl} for an example.
 \begin{figure}[h]
  \centering
  \includegraphics[width=14cm]{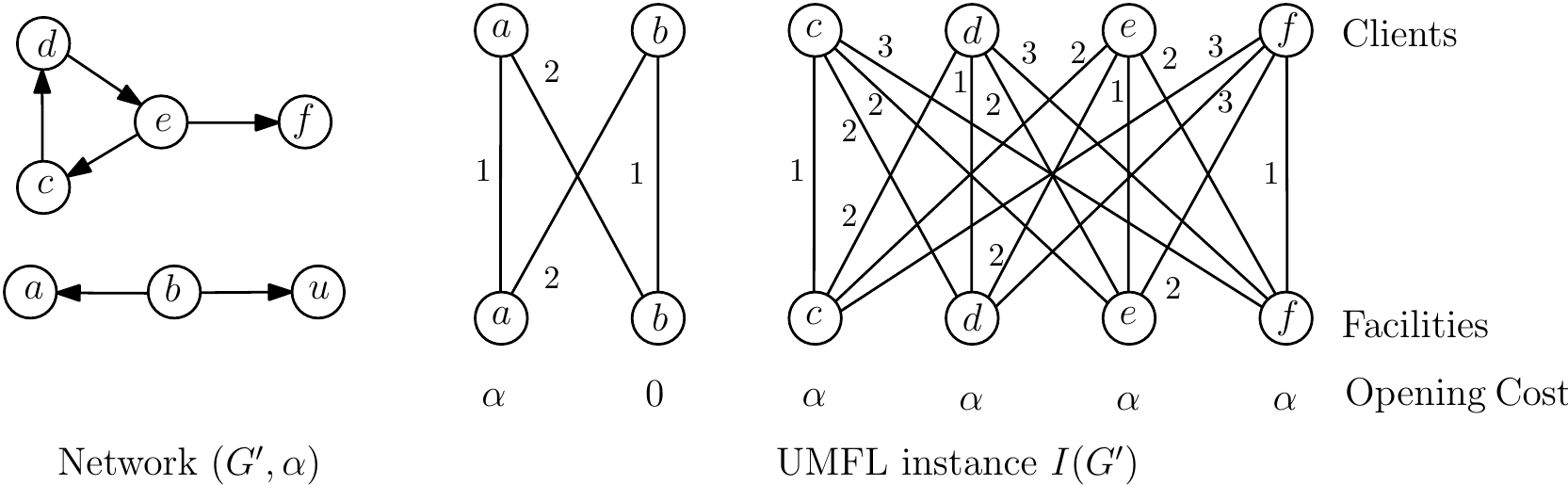}
  \caption{Network $(G',\alpha)$ and its corresponding UMFL instance~$I(G')$. Edges between clients and between facilities are omitted. All other omitted edges have length $\infty$.}
  \label{fig:umfl}
\end{figure}

  \noindent Now, observe that any strategy $S \in \mathcal{S}$ of agent~$u$ in $(G',\alpha)$ corresponds to the solution of the UMFL instance $I(G')$, where exactly the facilities in $F_S = S\cup Z$ are opened and where all clients are assigned to their nearest open facility. Moreover, every solution $F' = X \cup Z$, where $X\subseteq F\setminus Z$, for instance $I(G')$ corresponds to agent~$u$'s strategy $X \in \mathcal{S}$ in $(G',\alpha)$. Let $\mathcal{S}_\text{UMFL} = \{W \subseteq F \mid Z \subseteq W\}$ denote the set of all solutions to instance $I(G')$, which open at least all cost $0$ facilities. Hence, we have a bijection $\pi: \mathcal{S} \to \mathcal{S}_\text{UMFL}$, with $\pi(S) = S \cup Z$ and $\pi^{-1}(X) = X \setminus Z$. Let $\pi(S) = F_S$ and let $(G_S,\alpha)$ denote the network $(G',\alpha)$, where agent~$u$ has bought all edges towards vertices in $S$. Let $cost(F_S)$ denote the cost of the solution $F_S$ to instance~$I(G')$. We have that agent~$u$'s cost in $(G_S,\alpha)$ is equal to 
the cost of the corresponding UMFL solution $F_S$, since 
  \begin{align*}
   c_u(G_S,\alpha) &= \alpha |S|+ \sum_{w\in V(G_S)\setminus\{u\}}\big(1+ \min_{x\in S\cup Z} d_{G'}(x,w)\big) \\ 
            &= \alpha |S| + 0|Z| + \sum_{w \in V(G_S)\setminus\{u\}} \min_{x\in S\cup Z}d_{xw} \\
          &= \alpha |F_S\setminus Z| + 0|Z| + \sum_{w \in C} \min_{x\in F_S} d_{xw} \quad
          = cost(F_S). 
  \end{align*}

  \noindent We claim the following: If agent~$u$ plays strategy $S \in \mathcal{S}$ and cannot decrease her cost by buying, deleting or swapping \emph{one} edge in $(G_S,\alpha)$, then we have that the cost of the corresponding solution $F_S \in \mathcal{S}_{\text{UMFL}}$ to instance $I(G')$ cannot be strictly decreased by opening, closing or swapping \emph{one} facility.\\  
  Proving the above claim suffices to prove Theorem~\ref{thm_approx_upperbound}. This can be seen as follows: For UMFL, Arya et al.~\cite{Arya04} have already shown that the locality gap of UMFL is~$3$, that is, that any UMFL solution in which clients are assigned to their nearest open facility and which cannot be improved by opening, closing or swapping \emph{one} facility is a $3$-approximation of the optimum solution.\\
  By construction of $I(G')$, we have that every facility $z \in Z$ is the unique facility which is nearest to some client $w \in C$. Thus, we have that in any locally optimal and any globally optimal UMFL solution to $I(G')$ all cost $0$ facilities must be open, since otherwise such a solution can be improved by opening a cost $0$ facility. Hence, every locally or globally optimal solution to $I(G')$ has a corresponding strategy of agent~$u$ which yields the same cost. Using the claim and the result by Arya et al.~\cite{Arya04}, it follows that if agent~$u$ cannot decrease her cost by buying, deleting or swapping an edge in $(G_S,\alpha)$ then we have $c_u(G_S,\alpha) \leq 3 c_u(G_{S^*},\alpha)$, where $S^*$ is agent~$u$'s optimal (non-greedy) strategy in $(G',\alpha)$ and $(G_{S^*},\alpha)$ the network induced by $S^*$.\\  
  Now we prove the claim. Let $\pi(S) = F_S$. We have already shown that $c_u(G_S,\alpha) = cost(F_S)$. Furthermore, we have $Z \subseteq F_S$. We prove the contra-positive statement of the claim. Assume that solution $F_S$ can be improved by opening, closing or swapping \emph{one} facility. Let $F_S'$ be this locally improved solution and let $cost(F_S') < cost(F_S)$. Note, that $Z \subseteq F_S'$ must hold. This is true, since by construction of $I(G')$ closing a cost $0$ facility increases the cost of any solution to $I(G')$. Hence, no facility $z\in Z$ can be included in a closing or swapping operation. It follows that the strategy $S' := \pi^{-1}(F_S')$ exists. Observe, that $S = F_S \setminus Z$ and $S' = F_S' \setminus Z$ must differ by one element. Furthermore, by cost-equality, we have that $c_u(G_{S'},\alpha) = cost(F_S') < cost(F_S) = c_u(G_S,\alpha)$. Hence, agent~$u$ can buy, delete or swap one edge in $(G_S,\alpha)$ to decrease her cost.
\end{proof}

\section{The Quality of Max Greedy Equilibria}\label{sec_quality_max}
In this section, we discuss the stability of networks in \textsc{Max}-GE. We will start by showing that operations of buying, deleting and swapping edges each may have a strong non-local flavor. See Fig.~\ref{fig:nonlocal} for an illustration. 
\begin{lemma}\label{lem_non_local}
 For $k \geq 2$ there is a network $(G,\alpha)$, where an agent can decrease her cost by buying/deleting/swapping $k$ edges but not by buying/deleting/swapping $j<k$ edges.
\end{lemma}
\begin{figure}[!h]
  \centering
  \includegraphics[width=12.5cm]{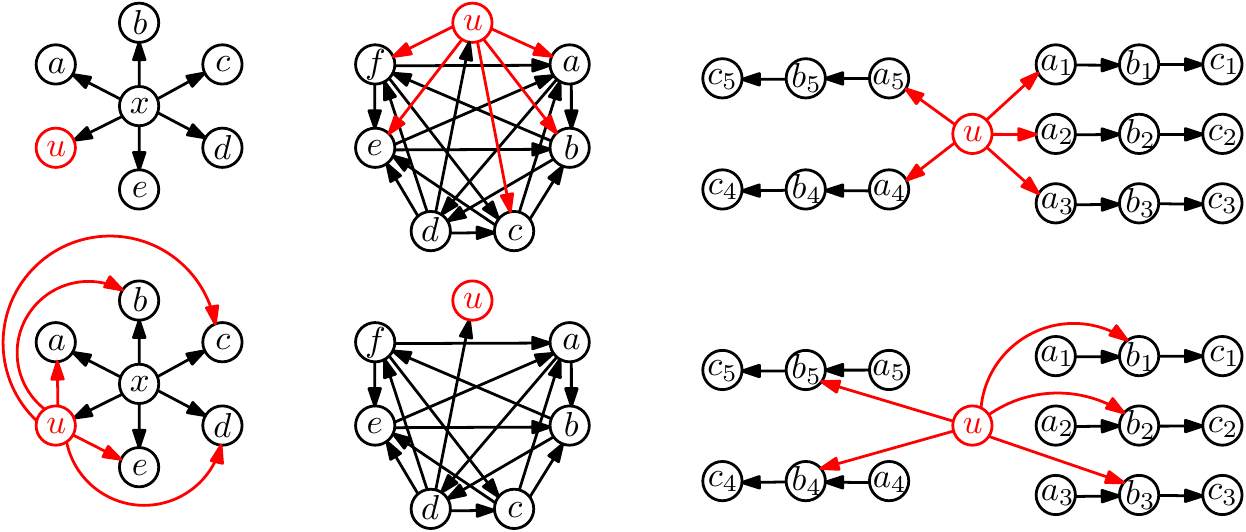}
  \caption{The networks and strategy-changes for $k=5$.}
  \label{fig:nonlocal}
\end{figure}
\begin{proof}
We consider each operation separately:
 \begin{itemize}
  \item Buying $k$ edges versus buying $j<k$ edges: Let $G$ be a star having $n = k+2$ vertices and let agent $u$ be a leaf vertex. See Fig.~\ref{fig:nonlocal}(left). Furthermore, let $\alpha < \frac{1}{n-2}$. Let $x$ be the center of the star. If agent $u$ owns the edge $\{u,x\}$, then we have $c_u(G) = \alpha + 2$, otherwise, we have $c_u(G) = 2$. Now, observe, that $u$ has exactly $k$ vertices at maximum distance $2$. Hence, buying $j<k$ edges does not decrease agent $u$'s maximum distance to any vertex. On the other hand, if $u$ buys $k = n-2$ edges to all distance $2$ vertices, then agent $u$'s distance cost decreases by $1$ while agent $u$'s edge cost increases by $k\alpha < 1$. Thus, buying $k$ edges yields a strict cost decrease for agent $u$.
  \item Deleting $k$ edges versus deleting $j<k$ edges: Let $G$ be a clique having $n = k+2$ vertices and let $u$ be an agent who owns all but one of her $k+1$ incident edges. See Fig.~\ref{fig:nonlocal}(middle). Let $\frac{1}{n-2} < \alpha \leq \frac{1}{n-3}$. Observe, that $c_u(G) = k\alpha +1$. If agent $u$ deletes $j<k$ edges, then $u$'s distance cost increases by $1$ while $u$'s edge cost decreases by $j\alpha \leq 1$. Thus, deleting $j<k$ edges does not decrease agent $u$'s cost. On the other hand, if $u$ deletes $k$ edges, then $u$ distance cost increases by $1$ while $u$'s edge cost decreases by $k\alpha > 1$. Hence, deleting $k$ edges decreases agent $u$'s cost.
  \item Swapping $k$ edges versus swapping $j<k$ edges: Let $G=(V,E)$ be a star-like graph which is defined as follows: Vertex $x$ is the center of the star and we have $k$ triples of vertices $x_i,y_i,z_i$, for $1 \leq i \leq k$. Let $E = \{(x,a_i),(a_i,b_i),(b_i,c_i)\mid 1\leq i \leq k\}$. See Fig.~\ref{fig:nonlocal}(right). Observe, that agent $x$ cannot decrease her cost by swapping $j<k$ edges simultaneously, since each edge must connect to the same subtree of $x$. In contrast to this, agent $x$ can decrease her cost by performing the multi-swap, where every edge $(x,a_i)$ is replaced by the edge $(x,b_i)$. Note, that this multi-swap decreases agent $x$'s distance-cost by $1$ while having the same edge cost.
 \end{itemize}
\end{proof}
\noindent Having seen Lemma~\ref{lem_non_local}, it should not come as a surprise that greedy local optimization may get stuck at sub-optimal states of the game. 

\subsection{Tree Networks in Max Greedy Equilibrium}
The examples on the left and right side of Fig.~\ref{fig:nonlocal} already show that there are tree networks, which are in \textsc{Max}-GE but not in \textsc{Max}-NE. In the following we show that this undesired behavior is restricted only to two families of tree networks in \textsc{Max}-GE. That is, we provide a characterization of all tree networks in \textsc{Max}-GE which are not in \textsc{Max}-NE. Furthermore, we show tight bounds on the stability for both mentioned families which are very close to the optimum. 
We start by introducing the main actors: \emph{Cheap Stars} and \emph{Badly Connected Trees}.
\begin{definition}[Cheap Star]
 A network $(T,\alpha)$ in \textsc{Max}-GE is called a \emph{Cheap Star}, if $T$ is a star having at least $n\geq 4$ vertices and $\alpha < \frac{1}{n-2}$. Furthermore, the ownership of all edges in $T$ is arbitrary.
\end{definition}
\begin{definition}[Badly Connected Tree]
 A tree network $(T,\alpha)$ in \textsc{Max}-GE is a \emph{Badly Connected Tree} if there is an agent $u\in V(T)$ who can decrease her cost by swapping $k>1$ own edges simultaneously. 
\end{definition}
\noindent Intuitively, Cheap Stars owe their instability to a multi-buy operation, whereas Badly Connected Trees owe their instability to a multi-swap operation. Observe that Cheap Stars have diameter $2$ and that Badly Connected Trees have diameter at least $3$. Hence, these families are disjunct. The following theorem shows that Cheap Stars and Badly Connected Trees are the \emph{only} tree networks in \textsc{Max}-GE which are not in \textsc{Max}-NE.
\begin{theorem}\label{thm_tree_char}
 Let $(T,\alpha)$ be a network in \textsc{Max}-GE, where $T$ is a tree. The network $(T,\alpha)$ is in \textsc{Max}-NE if and only if it is not a Cheap Star or a Badly Connected Tree. 
\end{theorem}
\noindent The proof of Theorem~\ref{thm_tree_char} is based on the following two observations. 
\begin{lemma}\label{lem_tree_diam2}
 Let $(T,\alpha)$ be a tree network in \textsc{Max}-GE having diameter at most~$2$. If $(T,\alpha)$ is not in \textsc{Max}-NE, then $(T,\alpha)$ is a Cheap Star.
\end{lemma}
\begin{lemma}\label{lem_tree_diam3}
 Let $(T,\alpha)$ be a tree network in \textsc{Max}-GE having diameter at least $3$. If $(T,\alpha)$ is not in \textsc{Max}-NE, then $(T,\alpha)$ is a Badly Connected Tree.
\end{lemma}
\noindent We start with proving Lemma~\ref{lem_tree_diam2}.
\begin{proof}[Proof of Lemma~\ref{lem_tree_diam2}]
  Trivially, any tree network having diameter at most $1$ must be in \textsc{Max}-NE. Hence, we focus on diameter $2$ tree networks which are in \textsc{Max}-GE but not in \textsc{Max}-NE. Note, that every tree network in \textsc{Max-GE} which has diameter $2$ is stable against multi-swap operations. This is easy to see, since leaves can own at most one edge and the unique non-leaf vertex (the center of the star) has already optimal distance cost of $1$. Since edge deletions lead to a disconnected network, it follows that the instability against arbitrary strategy-changes must be due to a multi-buy operation of a leaf agent. If $T$ has at most $3$ vertices, then any leaf can buy at most one additional edge, which represents a greedy operation. Hence, $T$ must have at least $4$ vertices. Since $T$ is a star, we have that any leaf $l$ has exactly $n-2$ vertices in distance~$2$. Thus, to strictly decrease her distance cost, agent $l$ must buy all edges towards these $n-2$ non-neighbors. It follows that such a 
multi-buy operation yields a strict cost decrease for agent $l$, if $\alpha < \frac{1}{n-2}$. This matches exactly the definition of a Cheap Star and implies that Cheap Stars are the only possible diameter $2$ tree networks in \textsc{Max}-GE which are not in \textsc{Max}-NE.
\end{proof}
\noindent For proving Lemma~\ref{lem_tree_diam3}, we first need some additional observations.
\begin{lemma}\label{lem_tree_diam3_alpha}
 If $(T,\alpha)$ is a tree network in \textsc{Max}-GE having diameter at least $3$, then $\alpha \geq 1$.
\end{lemma}
\begin{proof}
 Assume towards a contradiction that there is a tree network $(T,\alpha)$, which is in \textsc{Max}-GE and has diameter at least~$3$ and where $\alpha < 1$. There are two cases: 

 If every leaf of $T$ is a neighbor of a $1$-center of $T$, then, since $T$ has diameter at least~$3$, there must be two $1$-center vertices of $T$ and $T$ has diameter exactly $3$. It is easy to see that any tree can have at most two $1$-center vertices. Thus, we have that $T$ must be a ``double-star''. Let $x,y$ be the two $1$-center vertices of $T$ and let $l$ be a leaf which is a neighbor of~$x$. Since $y$ is a $1$-center, there must be a leaf $z$ which is a neighbor of $y$ and where $d_T(l,z) = 3$. If agent $l$ buys the edge $\{l,y\}$, then $l$'s edge cost increases by $\alpha$ but her distance cost decreases by $1$. Since $\alpha < 1$, this yields a strict cost decrease for agent $l$ and we have a contradiction to $(T,\alpha)$ being in \textsc{Max}-GE.

 If not all leaves of $T$ have a neighboring $1$-center vertex, then let $l$ be one such leaf which has the maximum distance to any $1$-center in $T$. Let $x$ be this $1$-center vertex and let $d_T(l,x) = k \geq 2$. Let $D_l$ be the set of vertices which have maximum distance to vertex $l$ in $T$. Since $l$ has maximum distance to $x$ and $x$ is a $1$-center of $T$, it follows that $x$ lies on all shortest paths from $l$ to any vertex in $D_l$. Thus, if agent $l$ buys the edge $\{l,x\}$, she reduces her distance cost by at least $1$ while increasing her edge cost by $\alpha <1$. This yields a strict cost decrease for agent $l$ and again we have a contradiction to $(T,\alpha)$ being in \textsc{Max}-GE.
\end{proof}

\begin{lemma}\label{lem_alpha_dist}
 Let $(G,\alpha)$ be any network in \textsc{Max}-GE which is not in \textsc{Max}-NE. Let $u$ be any player who can strictly decrease her cost by performing a non-greedy strategy-change towards strategy $S_u^*$. If $\alpha \geq 1$, then agent $u$'s distance cost induced by strategy $S_u^*$ is at least $2$.  
\end{lemma}
\begin{proof}
 Let $S_u$ be agent $u$'s current strategy in $(G,\alpha)$ and let $cost(u) = edge(u) + dist(u)$ denote agent $u$'s cost, where $edge(u)$ and $dist(u)$ denote $u$'s edge cost and distance cost in $(G,\alpha)$, respectively. Clearly, we have $dist(u)>1$, since otherwise $u$ can improve on her current strategy only by deleting edges, which yields an increase in distance cost by at least~$1$. We assume towards a contradiction that agent $u$ can perform a strategy-change towards strategy $S_u^*$, which yields a strict cost decrease and where $dist^*(u) =  1$. Here $dist^*(u)$ is agent $u$'s distance cost induced by strategy~$S_u^*$ and $cost^*(u) = edge^*(u)+dist^*(u)$ denotes $u$'s new cost. It is easy to see that $|S_u^*| > |S_u|$ must hold, since agent $u$ cannot possibly bring her distance cost down to $1$ by performing a multi-swap. Consider agent $u$'s shortest path in $(G,\alpha)$ towards a vertex having maximum distance to $u$. Clearly, this path has length $dist(u)$. Since $S_u^*$ yields $dist^*(u)=1$, 
it follows that $u$ must buy an edge to all vertices on this path. Hence, agent $u$ must buy at least $dist(u)-1$ many additional edges to achieve distance cost of $1$. But, since $\alpha \geq 1$, this yields that $edge^*(u)\geq edge(u)+(dist(u)-1)\alpha \geq edge(u)+dist(u)-1$. Hence, we have that $cost^*(u) = edge^*(u) + dist^*(u) \geq edge^*(u)+dist(u)-1 + 1 \geq edge(u)+dist(u) = cost(u)$, which is a contradiction to the fact that $S_u^*$ strictly decreases agent $u$'s cost. 
\end{proof}
\noindent Now we are ready for the proof of Lemma~\ref{lem_tree_diam3}.
\begin{proof}[Proof of Lemma~\ref{lem_tree_diam3}]
  Let $(T,\alpha)$ be any tree network in \textsc{Max}-GE, where $T$ has diameter at least $3$. Note, that by Lemma~\ref{lem_tree_diam3_alpha}, it follows that $\alpha \geq 1$. 

  We claim that if there is an agent $u$ in $V(T)$ with strategy $S_u$ who can strictly decrease her cost by changing to a strategy $S_u^*$, where $|S_u| < |S_u^*|$, then there must be a player $p$ who can strictly decrease her cost by buying \emph{one} edge. This yields a contradiction to $(T,\alpha)$ being in \textsc{Max}-GE. 

  Observe, that in a tree network, no player can change to a strategy which involves buying less edges than before, since such a change would disconnect the network. Proving the above claim suffices to prove the Lemma, since Badly Connected Trees are exactly those tree networks in \textsc{Max}-GE, where one agent $v$ with strategy $S_v$ can strictly decrease her cost by performing a multi-swap, that is, agent $v$ can change to a strategy $S_v^*$, where $|S_v| = |S_v^*|$. 

  Now we prove the claim. Let $u$ be an agent with strategy $S_u$ who can strictly decrease her cost by changing to strategy $S_u^*$, with $|S_u| < |S_u^*|$. Let $x_1,\dots,x_l$ denote the neighbors of $u$ in $T$ and let $T_{x_i}$ denote the maximal subtree rooted at $x_i$ which does not contain $u$, for all $1\leq i \leq l$. 
  Let $k = |S_u^*| - |S_u|$ denote the number of additional edges purchased by agent $u$ with her new strategy. Since we assume that $S_u^*$ yields a strict cost decrease for agent $u$, it follows that $S_u^*$ must decrease agent $u$'s distance cost by \emph{more than} $k\alpha$. Let $D_u$ denote the set of vertices of $T$ which have maximum distance to $u$ and let $dist(u)$ denote this distance. Furthermore, let $dist^*(u)$ denote agent $u$'s maximum distance induced by strategy $S_u^*$. By Lemma~\ref{lem_alpha_dist}, it follows that $2\leq dist^*(u) < dist(u)-k\alpha$, which yields $dist(u)\geq \lfloor k\alpha \rfloor +3$. There are two cases: 
  \begin{enumerate}
   \item $D_u \not\subset V(T_v)$, for any $v\in \{x_1,\dots,x_l\}$. In this case there are two vertices $p,q$, where $p \in V(T_{x_i})$ and $q \in V(T_{x_j})$, for some $i\neq j$, and we have $dist(u) = d_T(u,p) = d_T(u,q)$. Since $T$ is a tree, it follows that $d_T(p,q) = 2 dist(u)$ and that $q \in D_p$, where $D_p$ is the set of vertices of $T$ which have maximum distance to $p$. Observe, that $p$ has distance at most $2dist(u)-2$ to any other vertex in $T_{x_i}$. This implies that vertex $u$ lies on agent $p$'s shortest paths to any vertex of $D_p$. If agent $p$ buys the edge $\{p,u\}$, then agent $p$'s distance to all vertices which are not in $V(T_{x_i})$ decreases by $dist(u)-1 \geq \lfloor k\alpha \rfloor +2 > \alpha$. Furthermore, edge $\{p,u\}$ yields that agent $p$'s distance to any vertex in $V(T_{x_i})$ is at most $1 + dist(u)$. Since $dist(u)>k\alpha +2$ and $k\geq 1$ it follows that $1+dist(u) < 2dist(u)-\alpha$. Thus we have that edge $\{p,u\}$ increases agent $p$'s edge cost by $\alpha$ but 
at the same time it decreases her distance cost by more than $\alpha$, which is a contradiction to $(T,\alpha)$ being in \textsc{Max}-GE. 
   \item $D_u \subset V(T_v)$, for some $v\in \{x_1,\dots,x_l\}$. Consider agent $p \in D_u$, for which $d_T(u,p) = dist(u) \geq \lfloor k\alpha \rfloor + 3$ holds. Since $D_u \subset V(T_v)$ we have that $d_T(u,w) \leq dist(u)-1$, for all $w\in V(T)\setminus V(T_v)$. Hence, on the one hand we have that agent $p$'s maximum distance in $T$ to any vertex in $V(T_v)$ is at most $2dist(u)-2$. On the other hand, agent~$p$'s maximum distance in $T$ to any vertex in $V(T)\setminus V(T_v)$ is at most $2dist(u)-1$. If agent $p$ buys the edge $\{p,v\}$, then we have that her maximum distance to any vertex in $V(T_v)$ decreases by $dist(u)-2 \geq \lfloor k\alpha \rfloor + 1 > \alpha$. For any other vertex $q \in V(T)\setminus V(T_v)$, we have that edge $\{p,v\}$ yields a distance of at most $1+dist(u)$ between $p$ and $q$. Since $dist(u)>k\alpha +2$ and $k\geq 1$ we have $1+dist(u)<2dist(u)-1-\alpha$. Hence, edge $\{p,v\}$ increases agent $p$'s edge cost by $\alpha$ but it decreases agent $p$'s maximum distance by more 
than $\alpha$, which implies that $p$ can greedily buy the edge $\{p,v\}$ and thereby strictly decrease her cost. This is a contradiction to $(T,\alpha)$ being in \textsc{Max}-GE. 
  \end{enumerate}
\end{proof}
\noindent Finally, we can set out for proving Theorem~\ref{thm_tree_char}.
\begin{proof}[Proof of Theorem~\ref{thm_tree_char}]
 If a \textsc{Max}-GE tree network $(T,\alpha)$ is a Cheap Star, then by definition of a Cheap Star, there is a leaf-agent who can strictly decrease her cost by buying edges to all non-neighboring vertices. Clearly, this implies that a Cheap Star cannot be in \textsc{Max}-NE. Furthermore, by definition of a Badly Connected Tree, we have that in every such tree network, there is an agent who can strictly decrease her cost by swapping $k>1$ edges simultaneously, which implies that such networks are not in \textsc{Max}-NE. Hence, it remains to show that Cheap Stars and Badly Connected Trees are the only tree networks which can be in \textsc{Max}-GE and at the same time not in \textsc{Max}-NE.

 On the one hand, by Lemma~\ref{lem_tree_diam2}, we have that for \textsc{Max}-GE tree networks having at most diameter $2$ Cheap Stars are the only tree networks which are not in \textsc{Max}-NE. On the other hand, by Lemma~\ref{lem_tree_diam3}, it follows that among all \textsc{Max}-GE tree networks having diameter at least $3$ only Badly Connected Trees are not in \textsc{Max}-NE. Since this case distinction covers every possible diameter, the Theorem follows.
\end{proof}

\noindent We can use the characterization provided by Theorem~\ref{thm_tree_char} to ``circumvent'' the hardness of deciding whether a tree network is in \textsc{Max}-NE. 
\begin{theorem}\label{thm_tree_checking}
 For every tree network $(T,\alpha)$ it can be checked in $\mathcal{O}(n^4)$ many steps whether $(T,\alpha)$ is in \textsc{Max}-NE.
\end{theorem}
\begin{proof}
 We can check whether a tree network $(T,\alpha)$ is in \textsc{Max}-NE as follows: First, we compute whether $(T,\alpha)$ is in \textsc{Max}-GE. If this test fails, then, since \textsc{Max}-GEs are a super class of \textsc{Max}-NEs, we have that $(T,\alpha)$ is not in \textsc{Max}-NE. On the other hand, if $(T,\alpha)$ is in \textsc{Max}-GE, then we have to check whether $(T,\alpha)$ is a Cheap Star or a Badly Connected Tree. If $(T,\alpha)$ is not a Cheap Star and not a Badly Connected Tree, then, by Theorem~\ref{thm_tree_char}, we have that $(T,\alpha)$ must be in \textsc{Max}-NE. Otherwise, $(T,\alpha)$ is not in \textsc{Max}-NE.

 Computing whether $(T,\alpha)$ is in \textsc{Max}-GE can be done in $\mathcal{O}(n^4)$ steps by checking for every agent if she can strictly decrease her cost by either swapping or buying one own edge. We can neglect edge deletions since such an operation disconnects the network.  An agent may own $\Omega(n)$ may edges, which implies that at most $\mathcal{O}(n^2)$ many edge-swaps are possible. Computing the incurred cost of a strategy can be done in linear time by performing a modified breath first search of the tree network. Since there are $\mathcal{O}(n)$ many possible edge purchases per agent, it follows that we can check in $\mathcal{O}(n^3)$ steps, if an agent can decrease her cost by performing a greedy strategy change.  

 Checking if $(T,\alpha)$ is a Cheap Star is possible in $\mathcal{O}(n)$ steps, since we only have to compute the diameter of $T$ and checking if $n$ and $\alpha$ have the right size. Computing whether $(T,\alpha)$ is a Badly Connected Tree is more involved since we have to check if there is an agent who can perform a multi-swap to strictly decrease her cost. This can be done by computing for every agent $u$ the vertices having the maximum distance to $u$, checking if $u$ owns all edges towards the respective subtrees of $T$ and by computing the $1$-centers~\cite{KH79_1} of those subtrees. Finally, by checking if all edges towards subtrees which contain maximum distance vertices do not connect to a $1$-center of that subtree it can be decided whether agent $u$ can perform a multi-swap which decreases her cost. Computing the vertices having maximum distance to $u$ can be done by a breath first search. Furthermore, there are linear time algorithms for computing the $1$-center of an vertex-unweighted tree - 
see for example the work of Kariv and Hakimi~\cite{KH79_1}. Hence, we have that checking if agent $u$ can perform a multi-swap to decrease her cost can be done in $\mathcal{O}(n)$ steps. 

 In total this yields $\mathcal{O}(n^4)$ steps for deciding whether $(T,\alpha)$ is in \textsc{Max}-NE.   
\end{proof}

\noindent We are interested in the stability of tree networks in \textsc{Max}-GE. By Theorem~\ref{thm_tree_char}, we only have to analyze the stability of Cheap Stars and Badly Connected Trees to get bounds on the stability on any tree network in \textsc{Max}-GE.
\begin{lemma}\label{lem_approx_cheap_star}
 Every Cheap Star is in $2$-approximate \textsc{Max}-NE. Furthermore, this bound is tight.
\end{lemma}
\begin{proof}
 We consider any Cheap Star $(T,\alpha)$. Since Cheap Stars have at least $4$ vertices, we have that $V(T)$ consists of $x$, the center of the star, and at least $3$ leaves $v_1,v_2$ and $v_3$. Let the edge ownership be arbitrary. Analogously to the proof of Lemma~\ref{lem_non_local}, we have that no leaf agent of $T$ can buy one edge to decrease her cost. Let $S_{v_1}$ denote agent $v_1$'s strategy in $(T,\alpha)$. Let $S_{v_1}^*$ be $v_1$'s strategy which buys all edges towards all non-neighbors in $T$. We claim that $S_{v_1}^*$ is agent $v_1$'s best strategy. Since every Cheap Star is in \textsc{Max}-GE, we have that agent $v_1$ cannot delete or swap one edge to decrease her cost. Since she owns at most one edge in $(T,\alpha)$, this rules out all deletion and swapping operations. Analogously to the proof of Lemma~\ref{lem_non_local}, buying exactly one edge does not decrease player $x_1$'s cost either. Note, that since $\alpha<\frac{1}{n-2} < 1$, we have that no strategy which yields distance cost of $2$ 
can have strictly less cost than $S_{v_1}$. Hence, the claim follows. 

 Let $cost(v_1)$ and $cost^*(v_1)$ denote agent $v_1$'s cost induced by strategy $S_{v_1}$ and $S_{v_1}^*$, respectively. We have 
 $$ \lim_{\alpha \to 0} \frac{cost(v_1)}{cost^*(v_1)} = \lim_{\alpha \to 0}\frac{2}{(n-2)\alpha + 1} = \lim_{\alpha \to 0}\frac{\alpha + 2}{(n-2)\alpha +1} = 2.$$
 Thus, independently of the ownership of edge $\{v_1,x\}$, we have that the approximation ratio approaches $2$ as $\alpha$ tends to $0$. Clearly, this also represents a tight lower bound of $2$ on this ratio.
\end{proof}
\begin{lemma}\label{lem_approx_badly_connected}
 Every Badly Connected Tree is in $\frac{6}{5}$-approximate \textsc{Max}-NE. Furthermore, this bound is tight.
\end{lemma}
\begin{proof}
 Remember, that Badly Connected Trees are exactly those tree networks in \textsc{Max}-GE, where an agent $u$ can strictly decrease her cost by performing a multi-swap. Clearly, any multi-swap does not change agent $u$'s edge cost. Thus, to maximize the ratio between agent $u$'s cost in the Badly Connected Tree $(T,\alpha)$ and $u$'s cost after the best possible multi-swap, we have to consider a Badly Connected Tree, where agent $u$ can decrease her distance cost as much as possible. By definition of a Badly Connected Tree, we have that agent $u$ has at least two vertices $p$ and $q$ in maximum distance $dist(u)$ and we know that $p$ and $q$ lie in different subtrees of $u$. Observe that, since $T$ is a tree, agent $u$ owns exactly one edge towards each subtree which contains maximum distance vertices. To ensure connectedness of the network, agent $u$ must swap those edges only within their respective subtree. It follows that the best possible multi-swap connects to the middle vertex of the shortest paths to 
all maximum distance vertices. Let $dist^*(u)$ be agent $u$'s distance after her best possible multi-swap. It follows that $dist^*(u) \geq 1 + \left\lceil \frac{dist(u)-1}{2}\right\rceil \geq \left\lceil \frac{dist(u)}{2} \right\rceil$.

 Let $(T,\alpha)$ be a Badly Connected Tree, which contains an agent $u$, with $dist(u) = k$ and where $u$ owns $j\geq 2$ edges. We have
 $$\frac{cost(u)}{cost^*(u)} \leq \frac{j\alpha + k}{j\alpha + \left\lceil\frac{k}{2}\right\rceil}\leq \frac{2\alpha + k}{2\alpha + \left\lceil\frac{k}{2}\right\rceil}.$$
 Note, that this ratio is maximized for a Badly Connected Tree $(T^*,\alpha)$, where an agent $u\in V(T^*)$ can decrease her distance cost from $k$ to $\left\lceil\frac{k}{2}\right\rceil$ and where $\alpha$ is as small as possible. However, we cannot simply choose $\alpha = 1$ since we have to ensure that $(T^*,\alpha)$ remains in \textsc{Max}-GE.

 We explicitly construct $(T^*,\alpha)$, which will serve at the same time as lower and upper bound construction. Clearly, $T^*$ must consist of a path $P$ of length $2k$, where agent $u$ is the middle vertex of this path. Without loss of generality, we can choose an odd $k$. Furthermore, to avoid greedy edge swaps, we assume that all ownership-arcs are directed from $u$ towards the leaves of $P$. That is, for every edge $\{x,y\}$ in $P$ we have that $x$ owns $\{x,y\}$ if $x$ is closer to $u$ than $y$. Thus, on path $P$ we have that $u$ is the only agent who owns two edges. Now, observe, that $(P,\alpha)$ is already stable against greedy deletions and greedy swaps. No agent $x \in V(P)$ can swap any single edge to decrease her cost, since the only owned edge is by construction an edge which does not lie on all shortest paths from $x$ to the vertices having maximum distance to $x$. However, agents may improve their cost by buying one additional edge if $\alpha$ is small enough. To rule out this possibility, 
we consider a leaf agent $l$ and choose $\alpha$ in such a way that $l$ cannot decrease her cost by buying one edge. Note, that leaf agents have the largest distance cost in $P$, which implies that they are exactly those agents which are most susceptible to single edge purchases. Thus, if no leaf agent can decrease her cost by buying one edge, then no other agent can. 

 Let $l_1$ and $l_2$ be the to leaf agents of $P$. Observe, that $l_1$'s best possible additional edge connects to some vertex $z$ which lies on the path between $u$ and $l_2$. Hence, this edge decreases agent $l_1$'s distance cost by at least $k-1$. We choose $\alpha$ such that it will neutralize this decrease in distance cost. It follows that to minimize $\alpha$, we have to ensure that $z$ is as close as possible to $u$. We force $z$ towards $u$ by adding two branches to $P$ as follows: Let $p_1$ and $p_2$ denote the vertices which lie in the middle of the path from $u$ to $l_1$ and $l_2$, respectively. Thus, we have $d_P(u,p_1) = d_P(u,p_2) = \left\lceil \frac{k}{2}\right\rceil$. To finally obtain $T^*$, we connect both vertices $p_1,p_2$ to a path of length $\left\lceil \frac{k}{2}\right\rceil -1$, respectively. Again, the ownership on these paths resembles the edge ownership on $P$, that is, the respective vertex closer to $u$ owns the edge. Let $l_1'$ and $l_2'$ be the leaves of the newly attached 
paths. Note, that these new paths do not change agent $u$'s distance decrease. See Fig.~\ref{fig:badlyconnected} for an illustration. 
\begin{figure}[!h]
  \centering
  \includegraphics[width=11cm]{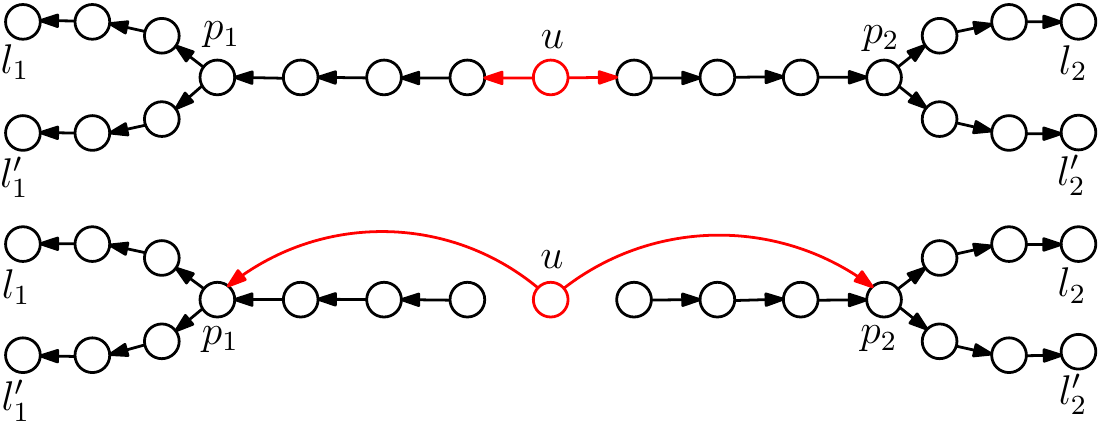}
  \caption{The network $(T^*,\alpha)$ for $k=7$ before and after agent $u$'s best multi-swap.}
  \label{fig:badlyconnected}
\end{figure}

 \noindent In $T^*$, we have that $d_{T^*}(l_1,l_1') = k - 1$. It follows that every possible additional edge of agent $l_1$ only yields a distance decrease of at most $k+1$. Thus, setting $\alpha \geq k+1$, implies that no agent in $T^*$ can decrease her cost by buying one edge. 

 Now we are ready to finally settle the approximation ratio of Badly Connected Trees. For agent $u$ in $(T^*,\alpha)$ we have
 $$ \frac{6}{5} = \lim_{k\to\infty}\frac{2(k-1)+k}{2(k-1)+\left\lceil\frac{k}{2}\right\rceil} \geq \lim_{k\to \infty}\frac{cost(u)}{cost^*(u)} \geq \lim_{k\to\infty}\frac{2(k+1)+k}{2(k+1)+\left\lceil\frac{k}{2}\right\rceil} = \frac{6}{5},$$
 where the limit on the left side represents the upper bound with $\alpha = k-1$ and the limit on the right represents the lower bound with $\alpha = k+1$. Both bounds match if we let $k$ tend to infinity. 
\end{proof}

\noindent Combining Theorem~\ref{thm_tree_char} with Lemma~\ref{lem_approx_cheap_star} and Lemma~\ref{lem_approx_badly_connected} we arrive at the following:
\begin{theorem}\label{thm_tree_approx}
Let $(T,\alpha)$ be a tree network in \textsc{Max}-GE. If $T$ has diameter at most $2$, then $(T,\alpha)$ is in $2$-approximate \textsc{Max}-NE. If $T$ has diameter at least $3$, then $(T,\alpha)$ is in $\frac{6}{5}$-approximate \textsc{Max}-NE. Moreover, both bounds are tight.
\end{theorem}

\subsection{Non-Tree Networks in Max Greedy Equilibrium}
Fig.~\ref{fig:nonlocal}~(middle) shows that there are non-tree networks in \textsc{Max}-GE, which are not in \textsc{Max}-NE. We want to quantify the loss in stability of \textsc{Max}-GEs versus \textsc{Max}-NEs. For tree networks we have that Cheap Stars play a crucial role. These networks owe their instability to a multi-buy operation and to the fact that they are in \textsc{Max}-GE for arbitrarily small $\alpha$. We generalize this property of Cheap Stars to non-tree networks.
\begin{definition}[Cheap Network]
 A network $(G,\alpha)$ in \textsc{Max}-GE, is called a \emph{Cheap Network}, if $(G,\alpha)$ remains in \textsc{Max}-GE when $\alpha$ tends to $0$. 
\end{definition}
\noindent Cheap Stars yield a lower bound on the stability approximation ratio which equals their diameter. We can generalize this observation:  
\begin{theorem}\label{thm_diam_approx}
 If there is Cheap Network $(G,\alpha)$ having diameter $d$, then there is an $\alpha^*$ such that the network $(G,\alpha^*)$ is in \textsc{Max}-GE but not in $\beta$-approximate \textsc{Max}-NE for any $\beta < d$.
\end{theorem} 
\begin{proof}
 Consider a Cheap Network $(G,\alpha)$, where $G$ has diameter $d$ and let $u$ be any vertex of $G$ having eccentricity $d$. Let $j$ denote the number of edges, which are owned by agent $u$ in $(G,\alpha)$. Thus, we have that agent $u$ has cost $j\alpha + d$. Now we consider the strategy change of agent $u$ towards the strategy which buys an edge to all vertices of $G$ which do not own an edge to $u$. Clearly, after the strategy change agent $u$ incurs cost at most $(n-1)\alpha + 1$. 

 Now, observe that since $(G,\alpha)$ is a Cheap Network, we have that $(G,\alpha)$ remains in \textsc{Max}-GE when $\alpha$ tends to $0$. Hence, $\lim_{\alpha\to 0}\frac{j\alpha + d}{(n-1)\alpha + 1} = d$, which implies that for all $\beta < d$ there is an $\alpha^*$ such that $j\alpha^* + d > \beta (n-1)\alpha^* + 1$. Hence, $(G,\alpha^*)$ is in \textsc{Max}-GE but not in $\beta$-approximate \textsc{Max}-NE for any $\beta < d$. 
\end{proof}
\begin{lemma}\label{lem_cheap_diam4}
 There is a Cheap Network having diameter $4$.
\end{lemma}
\begin{proof}
 We construct the Cheap Network $(\tilde{G},\alpha)$ as follows: 
 The graph $\tilde{G}$ has $24$ vertices $u_0,\dots,u_7,v_0,\dots,v_7,w_0,\dots,w_7$ and the vertices $u_0,\dots,u_7$ form a cycle, where $u_i$ owns the edge towards $u_{i+1}$, for $0\leq i \leq 6$, and $u_7$ owns the edge to $u_0$. Furthermore, for $0\leq j \leq 7$ we have that agent $v_j$ owns an edge to $u_j$ and $w_j$ and agent $w_j$ owns an edge towards $u_k$, where $k = (j + 4) \mod 8$. See Fig.~\ref{fig:diam4} for an illustration.
 \begin{figure}[!h]
  \centering
  \includegraphics[width=8cm]{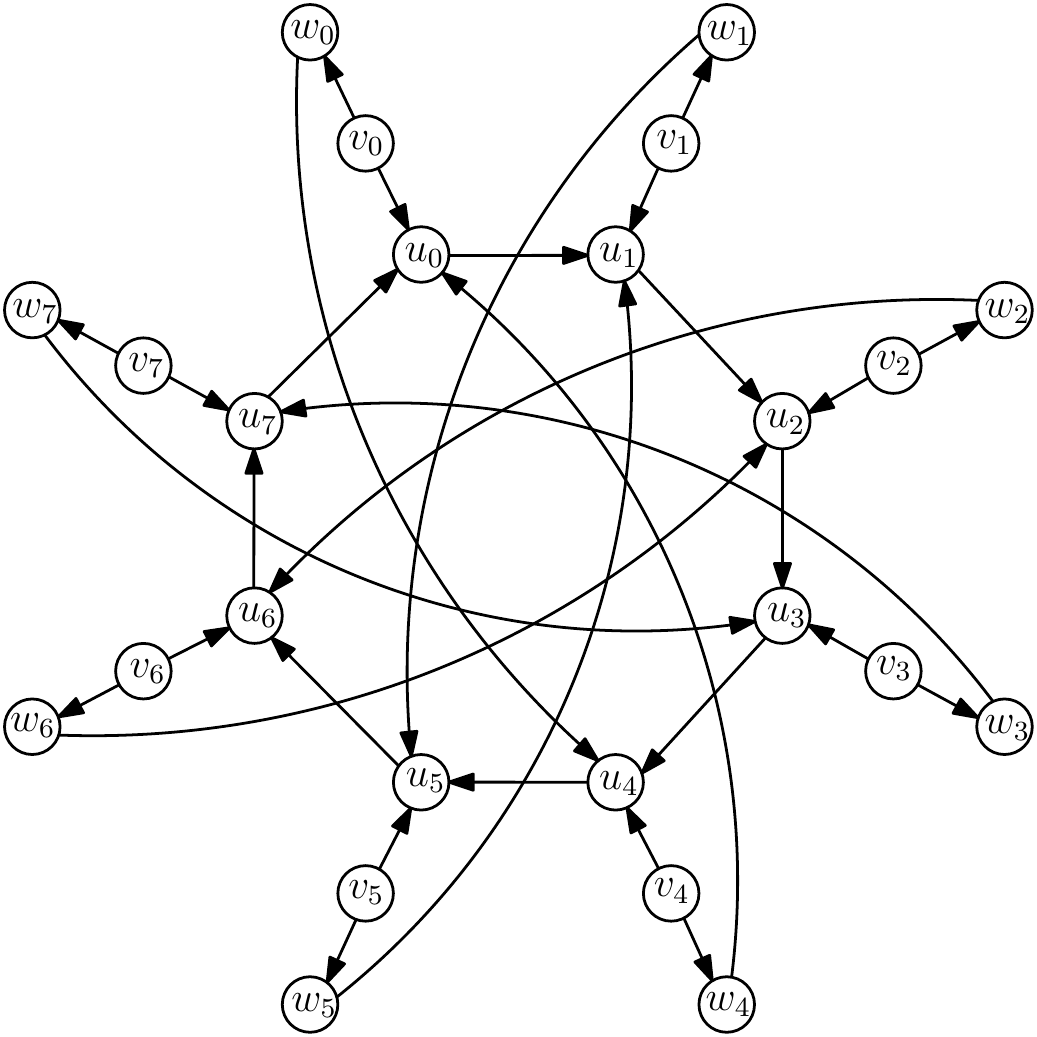}
  \caption{The Cheap Network $(\tilde{G},1)$ having diameter $4$.}
  \label{fig:diam4}
\end{figure}

 \noindent For showing that $(\tilde{G},\alpha)$ is a Cheap Network, we have to show that it is in \textsc{Max}-GE for some $\alpha$ and that it remains in \textsc{Max}-GE if $\alpha$ tends to $0$. We begin by proving that $(\tilde{G},\alpha)$ is in \textsc{Max}-GE for $\alpha = 1$.

 Since $(\tilde{G},1)$ is highly symmetric, it suffices to show that agents $u_0,v_0$ and $w_0$ cannot strictly decrease their cost by performing a greedy strategy change. It is easy to see that none of them can decrease her cost by deleting one own edge, since any such deletion increases the respective agent's distance cost by at least $1$. Since $\alpha = 1$, this does not yield a strict cost decrease. 

 Next, we show that none of the three agents can swap or buy an own edge and thereby strictly decrease her cost. Consider agent $u_0$, who owns exactly one edge. Now, observe, that $u_0$'s edge $\{u_0,u_1\}$ is the first edge on $u_0$'s shortest paths to the vertices $w_1$ and $v_5$ to which $u_0$ has maximum distance. Furthermore, observe, that there is no vertex in $\tilde{G}$ which is a neighbor to both $w_1$ and $v_5$. Thus, no swap can simultaneously decrease agent $u_0$'s distance to $w_1$ and $v_5$. Agent $u_0$ also has vertex $w_2$ in maximum distance $3$. Since $d_{\tilde{G}}(w_1,w_2) = 4$, it follows that $u_0$ cannot buy any edge, which strictly decreases $u_0$'s distances to both $w_1$ and $w_2$ simultaneously. Hence, $u_0$ cannot greedily purchase an edge to strictly decrease her cost. 

 Agent $v_0$ has, among others, vertices $v_2,v_3,v_5$ and $v_6$ in maximum distance $4$. We have $d_{\tilde{G}}(v_2,v_3) = d_{\tilde{G}}(v_5,v_6) = 3$ and $d_{\tilde{G}}(v_2,v_5) = d_{\tilde{G}}(v_3,v_6) = 4$. Thus, the best possible swap or edge purchase of $v_0$, which strictly decreases the distances from $v_0$ to $v_2$ and $v_3$ simultaneously, must connect to a neighbor $x$ of $v_2$ or $v_3$. It follows that either $d_{\tilde{G}}(x,v_5) \geq 3$ or $d_{\tilde{G}}(x,v_6) \geq 3$. Thus, such a swap or edge purchase does not reduce $v_0$'s distances to all of the four vertices $v_2,v_3,v_5,v_6$. Any improving swap or edge purchase must strictly decrease $v_0$'s distance cost, which implies that such an edge must connect to a neighbor of $v_2$ or $v_3$, since this is the only way to decrease the distance to both of them. It follows that $v_0$ cannot swap or buy any own edge to decrease her cost. 

 Agent $w_0$ has the vertices $v_1,v_2.v_6$ and $v_7$ in maximum distance $4$. We have $d_{\tilde{G}}(v_1,v_2) = d_{\tilde{G}}(v_6,v_7) = 3$ and $d_{\tilde{G}}(v_1,v_6) = d_{\tilde{G}}(v_2,v_7) = 4$. Hence, $w_0$ faces essentially the same situation as $v_0$ and an analogous argument shows that $w_0$ cannot swap or buy an edge to decrease her cost. This shows that $(\tilde{G},1)$ is in \textsc{Max}-GE.

 We have argued above that any edge deletion increases the distance cost of the moving agent by $1$. Furthermore, we have shown that no swap or edge purchase can strictly decrease any agent's distance cost. This implies that $(\tilde{G},\alpha)$ is in \textsc{Max}-GE for any $\alpha \leq 1$. Hence $(\tilde{G},\alpha)$ is a Cheap Network having diameter~$4$. 
\end{proof}
\begin{remark}
 The Cheap Network $(\tilde{G},\alpha)$ is not only stable against greedy strategy changes, it is even stable against any strategy change. That is, $(\tilde{G},\alpha)$ is in \textsc{Max}-NE for any $\alpha \leq 1$. To the best of our knowledge, this is the first known non-tree \textsc{Max}-NE network having diameter~$4$.
\end{remark}
\begin{corollary}
 For $\alpha < 1$ there is a network $(G,\alpha)$ in \textsc{Max}-GE, which is not in $\beta$-approximate \textsc{Max}-NE for any $\beta<4$.
\end{corollary}

\noindent Now we consider the case, where $\alpha \geq 1$. Quite surprisingly, it turns out that this case yields a very high lower bound on the approximation ratio. 
\begin{theorem}\label{thm_approx_nontree_lowerbound}
 For $\alpha \geq 1$ there is a \textsc{Max}-GE network $(G,\alpha)$ having $n$ vertices, which is not in $\beta$-approximate \textsc{Max}-NE for any $\beta<\frac{n-1}{5}$.
\end{theorem}
We give a family of networks in \textsc{Max}-GE each having an agent $u$ who can decrease her cost by a factor of $\frac{n-1}{5}$ by a non-greedy strategy-change. 
The network $(G_1,\alpha)$ can be obtained as follows: $V(G_1) =\{u,v,l_1,l_2,a_1,a_2,b_1,b_2,x_1,y_1\}$ and agent $u$ owns edges to $a_1$, $a_2$ and $x_1$. For $i \in \{1,2\}$, agent $b_i$ owns an edge to $v$ and to $a_i$ and agent $l_i$ owns an edge to $b_i$. Finally, agent $y_1$ owns an edge to $x_1$ and to $v$. Fig.~\ref{fig:lowerbound_max}~(left) provides an illustration. To get the $k$-th member of the family, for $k\geq 2$, we simply add the vertices $x_j,y_j$, for $2\leq j \leq k$, and let agent $y_j$ own edges towards $x_j$ and $v$. See Fig.~\ref{fig:lowerbound_max}~(right).
 \begin{figure}[!h]
  \centering
  \includegraphics[width=14cm]{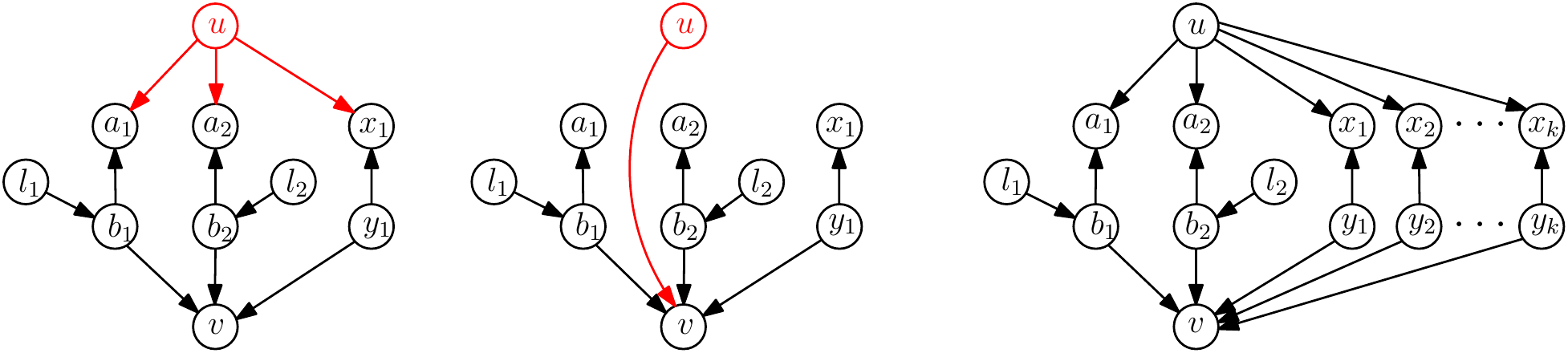}
  \caption{$(G_1,\alpha)$ before (left) and after (middle) agent $u$'s non-greedy strategy change and the network $(G_k,\alpha)$ (right).}
  \label{fig:lowerbound_max}
\end{figure}

\begin{lemma}\label{lem_lowerbound_stability}
 Each of the networks $(G_i,\alpha)$, as described above, is in \textsc{Max}-GE for $1\leq \alpha \leq 2$. 
\end{lemma}
\begin{proof}
 The statement is proven as follows: If any agent of $(G_1,\alpha)$ deletes one own edge, then either this operation disconnects the network or her distance cost increases by $2$. Since deleting an own edge decreases the edge cost by $\alpha \leq 2$, we have that such a move cannot yield a strict cost decrease for any agent. 

 Next, we show that no agent can swap an own edge to strictly decrease her cost. Clearly, agents $a_1,a_2,v$ and $x_1,\dots,x_k$ cannot swap any edge since they do not own one. By symmetry of the construction, we only have to show that agents $u,b_1,l_1$ and $y_1$ cannot decrease their cost by swapping one own edge. Agent $u$ has vertices $l_1,l_2$ and $v$ in maximum distance~$3$. But since $d_{G_k}(l_1,l_2) = 4$ it is impossible for $u$ to swap any own edge such that the distance to \emph{all} three of them is strictly decreased. Agent $b_1$ has vertices $l_2,a_2$ and $x_1$ in maximum distance $3$. But since $d_{G_k}(l_2,x_1) = 4$, no swap can decrease $b_1$'s distance to all of them. Analogously, the same holds true for agent $y_1$, who has $l_1,l_2,a_1$ and $a_2$ in maximum distance $3$. Agent $l_1$ cannot improve by swapping her edge, since $b_1$ is a $1$-center vertex of the graph $G_k - l_1$ and for a leaf vertex it is clearly optimal to connect to a $1$-center of the remaining network.

 Finally, let us focus on greedy edge purchases in $(G_k,\alpha)$. Since $\alpha \geq 1$, it follows that greedily buying one edge can strictly decrease an agent's cost only if this operation decreases the distance cost of that agent by more than $\alpha$, that is, by at least $2$. Clearly, for all agents of $(G_k,\alpha)$ which have eccentricity $3$, this is impossible. Now we consider all other agents, which all have eccentricity $4$ and we show that none of them can buy an edge to decrease her distance cost by more than $1$. By symmetry, it suffices to argue for agents $l_1$, $a_1$ and $x_1$. Agent $l_1$ has vertices $l_2$ and $x_1$ in maximum distance $4$. To decrease both distances simultaneously, agent $l_1$ must buy an edge towards a vertex, which lies on both shortest paths from $l_1$ to $l_2$ and from $l_1$ to $x_1$. Indeed, vertex $v$ is such a vertex and it is easy to see that it is the only non-neighbor of $l_1$, which lies on both shortest paths. But buying an edge towards $v$ decreases $l_1$'s 
distance cost only by $1$. Agent $a_1$ only has vertex $l_2$ in maximum distance $4$. There are two shortest paths from $a_1$ to $l_2$ which both use vertex $b_2$. However, buying an edge towards $b_2$ only yields a distance decrease of $1$ for agent $a_1$. The same holds true for an edge towards $a_2$ or $v$, respectively. Thus, no edge to any non-neighboring vertex on $a_1$'s shortest paths to $l_2$ can decrease $a_1$'s distance cost by more than $1$. Agent $x_1$ has vertices $l_1$ and $l_2$ in maximum distance $4$. But, analogously to agent $l_1$'s situation, there is no vertex which simultaneously lies on a shortest path from $x_1$ to $l_1$ and on a shortest path from $x_1$ to $l_2$ and which has distance $1$ to $l_1$ and $l_2$. Thus, agent $x_1$ can decrease her distance cost by buying one edge by at most~$1$. 
\end{proof}

\begin{proof}[Proof of Theorem~\ref{thm_approx_nontree_lowerbound}]
 We focus on agent $u$ in the network $(G_k,\alpha)$ and show that this agent can change her strategy in a non-greedy way and thereby decrease her cost by a factor of $\frac{n-1}{5}$, where $n$ is the number of vertices of $G_k$. Let $S_u$ be agent $u$'s current strategy in $(G_k,\alpha)$ and let $S_u^*$ be $u$'s strategy which only buys one edge towards vertex $v$. See Fig~\ref{fig:lowerbound_max}~(left and middle). Let $cost(u)$ and $cost^*(u)$ denote agent $u$'s cost induced by strategy $S_u$ and $S_u^*$, respectively. For $\alpha = 2$, we have
 $$\frac{cost(u)}{cost^*(u)} = \frac{\alpha(2+k)+3}{\alpha + 3} = \frac{7}{5} + \frac{2k}{5} = \frac{n-1}{5},$$ where the last equality follows since $k = \frac{n-8}{2}$, by construction.
\end{proof}

\begin{corollary}\label{cor_minmaxFL}
 Uncapacitated Metric Min-Max Facility Location has a locality gap of $\frac{n-1}{5}$, where $n$ is the number of clients.
\end{corollary}
\begin{proof}
 The corollary follows by using the ``locality gap preserving'' reduction provided in the proof of Theorem~\ref{thm_approx_upperbound} and the lower bound of Theorem~\ref{thm_approx_nontree_lowerbound}.

 The lower bound construction of Theorem~\ref{thm_approx_nontree_lowerbound} can be transformed into an instance of uncapacitated metric min-max facility location. Remember, that we have cost-equality and that greedy strategy-changes of agent $u$ in the NCG transfer one to one to greedy modifications of the facility location solution. Thus, we have the property that the corresponding solution to the facility location problem is locally optimal but resembles only a $\frac{n-1}{5}$-approximation to the globally optimal solution. 
\end{proof}

\subsection*{Acknowledgments}
I am grateful to Achim Passen for many interesting discussions and helpful comments. Furthermore, I thank an anonymous referee for pointing to \cite{GKSB12}.  
\bibliographystyle{abbrv}
\bibliography{lenzner_greedy_selfish_network_creation}

\newcommand{\SortNoop}[1]{}
\begin{thebibliography}{10}

\bibitem{Al06}
S.~Albers, S.~Eilts, E.~Even-Dar, Y.~Mansour, and L.~Roditty.
\newblock On nash equilibria for a network creation game.
\newblock In {\em Proceedings of the 17th annual ACM-SIAM symposium on Discrete
  algorithm}, SODA '06, pages 89--98, New York, NY, USA, 2006. ACM.

\bibitem{AL10}
S.~Albers and P.~Lenzner.
\newblock On approximate nash equilibria in network design.
\newblock In A.~Saberi, editor, {\em Internet and Network Economics}, volume
  6484 of {\em LNCS}, pages 14--25. Springer Berlin / Heidelberg, 2010.

\bibitem{ADHL10}
N.~Alon, E.~D. Demaine, M.~Hajiaghayi, and T.~Leighton.
\newblock Basic network creation games.
\newblock In {\em Proceedings of the 22nd ACM symposium on Parallelism in
  algorithms and architectures}, SPAA '10, pages 106--113, New York, NY, USA,
  2010. ACM.

\bibitem{Arya04}
V.~Arya, N.~Garg, R.~Khandekar, A.~Meyerson, K.~Munagala, and V.~Pandit.
\newblock Local search heuristics for k-median and facility location problems.
\newblock {\em SIAM J. Comput.}, 33(3):544--562, 2004.

\bibitem{CHKS12}
A.~Cord-Landwehr, M.~H\"ullmann, P.~Kling, and A.~Setzer.
\newblock Basic network creation games with communication interests.
\newblock In {\em SAGT}, 2012, to appear.

\bibitem{De09}
E.~D. Demaine, M.~Hajiaghayi, H.~Mahini, and M.~Zadimoghaddam.
\newblock The price of anarchy in cooperative network creation games.
\newblock {\em SIGecom Exch.}, 8(2):2:1--2:20, Dec. 2009.

\bibitem{De07}
E.~D. Demaine, M.~T. Hajiaghayi, H.~Mahini, and M.~Zadimoghaddam.
\newblock The price of anarchy in network creation games.
\newblock {\em ACM Trans. on Algorithms}, 8(2):13, 2012.

\bibitem{Ehs11}
S.~Ehsani, M.~Fazli, A.~Mehrabian, S.~Sadeghian~Sadeghabad, M.~Safari,
  M.~Saghafian, and S.~ShokatFadaee.
\newblock On a bounded budget network creation game.
\newblock In {\em Proceedings of the 23rd ACM symposium on Parallelism in
  algorithms and architectures}, SPAA '11, pages 207--214, New York, NY, USA,
  2011. ACM.

\bibitem{Fab03}
A.~Fabrikant, A.~Luthra, E.~Maneva, C.~H. Papadimitriou, and S.~Shenker.
\newblock On a network creation game.
\newblock In {\em Proc. of the 22nd annual symp. on Principles of distributed
  computing}, PODC '03, pages 347--351, New York, NY, USA, 2003. ACM.

\bibitem{GKSB12}
A.~Guly{\'a}s, A.~K{\~o}r{\"o}si, D.~Szab{\'o}, and G.~Bicz{\'o}k.
\newblock On greedy network formation.
\newblock In {\em Proceedings of ACM SIGMETRICS/Performance W-PIN}, 2012.

\bibitem{KH79_1}
O.~Kariv and S.~L. Hakimi.
\newblock An algorithmic approach to network location problems. i: The
  p-centers.
\newblock {\em SIAM J. on Appl. Math.}, 37(3):pp. 513--538, 1979.

\bibitem{KH79}
O.~Kariv and S.~L. Hakimi.
\newblock An algorithmic approach to network location problems. ii: The
  p-medians.
\newblock {\em SIAM J. on Appl. Math.}, 37(3):pp. 539--560, 1979.

\bibitem{L11}
P.~Lenzner.
\newblock On dynamics in basic network creation games.
\newblock In G.~Persiano, editor, {\em Algorithmic Game Theory}, volume 6982 of
  {\em LNCS}, pages 254--265. Springer Berlin / Heidelberg, 2011.

\bibitem{L12}
P.~Lenzner.
\newblock Greedy selfish network creation.
\newblock In P.~W. Goldberg and M.~Guo, editors, {\em Internet and Network
  Economics}, volume 7695 of {\em LNCS}. Springer Berlin / Heidelberg, 2012, to
  appear.

\bibitem{MS10}
M.~Mihal{\'a}k and J.~Schlegel.
\newblock The price of anarchy in network creation games is (mostly) constant.
\newblock In S.~Kontogiannis, E.~Koutsoupias, and P.~Spirakis, editors, {\em
  Algorithmic Game Theory}, volume 6386 of {\em LNCS}, pages 276--287. Springer
  Berlin / Heidelberg, 2010.

\bibitem{MS12}
M.~Mihal{\'a}k and J.~Schlegel.
\newblock Asymmetric swap-equilibrium: A unifying equilibrium concept for
  network creation games.
\newblock In B.~Rovan, V.~Sassone, and P.~Widmayer, editors, {\em Mathematical
  Foundations of Computer Science 2012}, volume 7464 of {\em LNCS}, pages
  693--704. Springer Berlin / Heidelberg, 2012.

\bibitem{Nash}
J.~F. Nash.
\newblock Equilibrium points in n-person games.
\newblock {\em PNAS}, 36(1):48--49, 1950.

\bibitem{Vaz01}
V.~V. Vazirani.
\newblock {\em Approximation algorithms}.
\newblock Springer, 2001.

\end{thebibliography}

\end{document}